%% file: main.tex
\documentclass[11pt]{article}

\usepackage[letterpaper,margin=1in]{geometry}

\usepackage{amsmath,amssymb,amsthm,thmtools,amsfonts,latexsym,bbm,xspace}
\usepackage{graphicx,float,mathtools,braket}
\usepackage{enumitem,booktabs,forest,soul,comment}
\usepackage{ninecolors}
\definecolor{citegreen}{HTML}{208054}
\definecolor{citeblue}{HTML}{0055cc}

\usepackage[pagebackref,colorlinks,bookmarks=true,hypertexnames=false]{hyperref}
\NineColors{saturation=high}
\hypersetup{
    breaklinks=true,   %
    colorlinks=true, %
    linkcolor=blue3, %
    citecolor=green3, %
    urlcolor=blue3, %
    hypertexnames=false,
}

\usepackage[nameinlink,capitalize]{cleveref}

\renewcommand{\backref}[1]{}

\renewcommand{\backrefalt}[4]{%
\ifcase #1 %
\or
[p.\ #2]%
\else
[pp.\ #2]%
\fi}

\newcommand{\indic}{\mathbf{1}}   %

\usepackage{mathdots} %

\usepackage[useregional]{datetime2}
\DTMusemodule{english}{en-US}

\usepackage{thm-restate}

\newtheorem{theorem}{Theorem}[section]

\newtheorem{lemma}[theorem]{Lemma}

\newtheorem{proposition}[theorem]{Proposition}
\newtheorem{corollary}[theorem]{Corollary}

\newtheorem{fact}[theorem]{Fact}
\crefname{fact}{Fact}{Facts}

\theoremstyle{definition}
\newtheorem{definition}[theorem]{Definition}

\newtheorem{remark}[theorem]{Remark}

\newcommand{\eps}{\epsilon}
\newcommand{\td}{d_{\mathrm{tr}}}

\renewcommand{\Pr}{\mathop{\bf Pr\/}}
\newcommand{\E}{\mathop{\bf E\/}}

\newcommand{\ketbra}[2]{\ket{#1}\!\!\bra{#2}}

\newcommand{\tr}{\mathrm{tr}}

\newcommand{\poly}{\mathrm{poly}}

\newcommand{\spn}{\mathrm{span}}

\newcommand{\yes}{{\mathrm{yes}}}
\newcommand{\no}{{\mathrm{no}}}

\newcommand{\R}{\mathbb{R}}

\newcommand{\Ayes}{A_{\mathrm{yes}}}
\newcommand{\Ano}{A_{\mathrm{no}}}

\newcommand{\Oyes}{O_{\mathrm{yes}}}
\newcommand{\Ono}{O_{\mathrm{no}}}

\newcommand{\lmin}{\lambda_{\mathrm{min}}}
\newcommand{\lmax}{\lambda_{\mathrm{max}}}

\newcommand{\pureSuperQMA}{\mathsf{pureSuperQMA}}

\newcommand{\reals}{\mathbb R}
\newcommand{\complex}{\mathbb C}
\newcommand{\CC}{\complex}
\newcommand{\RR}{\reals}
\newcommand{\nats}{\mathbb N}
\newcommand{\NN}{\nats}

\newcommand{\class}[1]{\mathsf{#1}}
\mathchardef\mhyphen="2D %

\newcommand{\NP}{\mathsf{NP}}
\newcommand{\coNP}{\class{coNP}}

\newcommand{\PSPACE}{\mathsf{PSPACE}}
\newcommand{\BPP}{\mathsf{BPP}}

\newcommand{\NEXP}{\mathsf{NEXP}}

\newcommand{\BQP}{\mathsf{BQP}}

\newcommand{\QMA}{\mathsf{QMA}}
\newcommand{\coQMA}{\mathsf{coQMA}}
\newcommand{\QMAtwo}{\mathsf{QMA(2)}}

\newcommand{\PH}{\mathsf{PH}}

\newcommand{\QPH}{\mathsf{QPH}}
\newcommand{\QSigmai}[1][i]{\mathsf{Q\Sigma_{#1}}}
\newcommand{\QPii}[1][i]{\mathsf{Q\Pi_{#1}}}

\newcommand{\pureQPH}{\mathsf{pureQPH}}
\newcommand{\pureQSigmai}[1][i]{\mathsf{pureQ\Sigma_{#1}}}
\newcommand{\pureQPii}[1][i]{\mathsf{pureQ\Pi_{#1}}}

\newcommand{\PSHSigmai}[1][i]{\mathsf{PSH\mhyphen\Sigma}_{#1}}
\newcommand{\PSHPii}[1][i]{\mathsf{PSH\mhyphen\Pi}_{#1}}

\newcommand{\MSHSigmai}[1][i]{\mathsf{MSH\mhyphen\Sigma}_{#1}}
\newcommand{\MSHPii}[1][i]{\mathsf{MSH\mhyphen\Pi}_{#1}}

\newcommand{\calA}{\mathcal{A}}
\newcommand{\calB}{\mathcal{B}}
\newcommand{\calC}{\mathcal{C}}
\newcommand{\calD}{\mathcal{D}}

\newcommand{\calH}{\mathcal{H}}

\newcommand{\calS}{\mathcal{S}}

\DeclarePairedDelimiter\abs{\lvert}{\rvert}
\DeclarePairedDelimiter\pars{\lparen}{\rparen}
\DeclarePairedDelimiter\maxnorm{\lVert}{\rVert_{\mathrm{max}}}
\newcommand{\norm}[1]{\left\lVert #1 \right\rVert}

\newcommand{\Pprod}{P_{\mathrm{prod}}}
\newcommand{\Pswap}{P_{\mathrm{swap}}}
\newcommand{\Piprod}{\Pi_{\mathrm{prod}}}
\newcommand{\Piswap}{\Pi_{\mathrm{swap}}}
\newcommand{\Pisym}{\Pi_{\mathrm{sym}}}

\newcommand{\pacc}{p_{\mathrm{acc}}}

\newcommand{\Nacc}{N_{\mathrm{acc}}}

\newcommand{\wteta}{\widetilde{\eta}}

\let\epsilon\varepsilon

\title{
On the Pure Quantum Polynomial Hierarchy\\and Quantified Hamiltonian Complexity
}
\author{Sabee Grewal\thanks{\texttt{sabee@cs.utexas.edu}. The University of Texas at Austin.} 
\and{Dorian Rudolph\thanks{\texttt{dorian.rudolph@upb.de}. Paderborn University and PhoQS.}}
}
\date{}

\begin{document}

\maketitle

\begin{abstract}
We prove several new results concerning the pure quantum polynomial hierarchy $\pureQPH$. 
First, we show that $\QMA(2) \subseteq \pureQSigmai[2]$, i.e., two unentangled existential provers can be simulated by competing existential and universal provers. 
We further prove that $\pureQSigmai[2]\subseteq \QSigmai[3] \subseteq \NEXP$. 
Second, we give an error reduction result for 
$\pureQPH$, and, as a consequence, prove that 
$\pureQPH = \QPH$. 
A key ingredient in this result is an improved dimension-independent disentangler.
Finally, we initiate the study of quantified Hamiltonian complexity, the quantum analogue of quantified Boolean formulae. 
We prove that the quantified pure sparse Hamiltonian problem is $\pureQSigmai$-complete. 
By contrast, other natural variants (pure/local, mixed/local, and mixed/sparse) admit nontrivial containments but fail to be complete under known techniques. 
For example, we show that the $\exists\forall$-mixed local Hamiltonian problem lies in $\NP^\QMA \cap \coNP^{\QMA}$.
\end{abstract}

\hypersetup{linktocpage}
\tableofcontents

\newpage
\section{Introduction}
The polynomial hierarchy ($\PH$) \cite{stockmeyerPolynomialtimeHierarchy1976} plays a central role in complexity theory. 
It has been instrumental in understanding the power of computational models such as
$\BPP$~\cite{m.sipserComplexityTheoreticApproach1983,c.lautemannBPPPolynomialTime1983}, low-depth classical circuits~\cite{furstParityCircuitsPolynomialtime1984}, counting classes~\cite{todaPPHardPolynomialTime1991}, and non-uniform computation~\cite{karplipton}. 
More recently, $\PH$ has been used to provide evidence for the hardness of simulating quantum circuits and has underpinned theoretical foundations for quantum supremacy demonstrations~\cite{bremnerClassicalSimulationCommuting2010,aaronsonComputationalComplexityLinear2011,boulandComplexityVerificationQuantum2019}.

Quantum generalizations of the polynomial hierarchy have been explored intermittently over the past two decades~\cite{t.yamakamiQuantumNPQuantum2002,j.lockhartQuantumStateIsomorphism2017,gharibian2022quantum,falor2023collapsiblepolynomialhierarchypromise,grewal_et_al:LIPIcs.CCC.2024.6,agarwal_et_al:LIPIcs.MFCS.2024.7}, and have recently begun to find broader applications within quantum complexity theory~\cite{agarwal2025cautionarynotequantumoracles}. 
Given the central role of $\PH$, it is natural to study its quantum analogues and their role in quantum complexity theory.

As with most prior work, we focus on the quantifier-based definitions of the quantum polynomial hierarchy. In this setting, the $i$th level, denoted $\QSigmai$, consists of promise problems that can be decided by a quantum polynomial-time verifier interacting with $i$ rounds of quantum proofs, where the quantifiers alternate between existential and universal. 
More concretely, $\QSigmai$ (resp. $\QPii$) consists of problems where the interaction begins with an existential (resp. universal) quantum proof, followed by alternating quantifiers over polynomial-size quantum states, with the verifier required to accept with high probability in the YES case and reject in the NO case.
The union over all levels defines the quantum polynomial hierarchy $\QPH$. The \emph{pure} quantum polynomial hierarchy $\pureQPH$ is defined analogously, except that the quantified quantum proofs are restricted to be pure states.\footnote{There is also a third natural variant: the \emph{entangled} quantum polynomial hierarchy, where the provers are allowed to entangle their proofs across rounds. Grewal and Yirka \cite{grewal_et_al:LIPIcs.CCC.2024.6} showed that this variant collapses to its second level.}

A central challenge in defining a quantum analogue of the polynomial hierarchy is determining the ``right'' formulation among several possible variants. This question has been raised before~\cite{gharibian2022quantum,agarwal_et_al:LIPIcs.MFCS.2024.7}, but it remains unresolved. This work tackles the problem directly: we prove several new results about $\pureQPH$ that, taken together, provide compelling evidence that it is the most natural quantifier-based definition of the quantum polynomial hierarchy.

At the same time, each of our results is interesting in its own right. 
Our first contribution gives a new upper bound on $\QMA(2)$, placing it in the second level of $\pureQPH$. This shows that two existential provers can be simulated by competing provers and recasts $\QMA(2)$ as a min-max optimization problem (rather than a nonconvex optimization over separable states). Our second and most technical result is an error reduction procedure for $\pureQPH$, which in turn implies $\pureQPH = \QPH$; the key tool is a new dimension-independent disentangler, extending recent work of Jeronimo and Wu~\cite{JW24}. 
Finally, we initiate the study of quantified Hamiltonian complexity, a quantum analogue of quantified Boolean formulae. These problems capture robust ground-state questions, such as whether there exists a state on one subsystem that ensures the overall system remains low-energy regardless of perturbations to the rest.

\subsection{Our Results}

Our first result establishes that $\QMAtwo \subseteq \pureQSigmai[2] \subseteq \QSigmai[3]$, i.e., that the second level of $\pureQPH$ is sandwiched between $\QMAtwo$ and the third level of $\QPH$. 
Prior work has shown that $\QSigmai[3] \subseteq \NEXP$ \cite{gharibian2022quantum}. 

\begin{theorem}[Combination of \cref{thm:qma2-in-psigma2,thm:pqsigma2-in-qsigma3}]\label{thm:intro-sandwich}
  $\QMA(2)\subseteq\pureQSigmai[2]\subseteq\QSigmai[3] \subseteq \NEXP$.
\end{theorem}

Informally, $\QMAtwo$ is the class of promise problems decidable given two unentangled quantum proofs. Since its introduction in 2001~\cite{kobayashi2001quantumcertificateverificationsingle}, it has been known that $\QMA \subseteq \QMAtwo \subseteq \NEXP$. However, whether $\QMAtwo$ is ``closer'' to $\QMA$ or to $\NEXP$ remains a central open problem in quantum complexity theory.

What \cref{thm:intro-sandwich} contributes to this question is nuanced. 
Our result shows that $\QMAtwo$ can be captured by a one-round quantum refereed-game model with an existential prover followed by an (adversarial) universal prover who has perfect knowledge of the first message.
Conceptually, this recasts the nonconvex ``maximize over separable witnesses'' view of $\QMAtwo$ as a single alternation
\[
\max_{\ket{\psi}} \min_{\ket\phi} \tr\left(\Pi(\ketbra{\psi}{\psi} \otimes \ketbra{\phi}{\phi})\right),
\]
that is, a saddle-point optimization problem. 
This perspective opens the door to applying techniques from min–max optimization, game theory, and the study of quantum refereed games to better understand the true power of $\QMAtwo$.
Moreover, in closely related models where provers are allowed \emph{mixed-state strategies}, the game value can be approximated in $\PSPACE$~\cite{jain2009parallel}. We view this as qualitative evidence that $\QMAtwo$ is perhaps not equal to $\NEXP$.

Our second result is an error reduction result for $\pureQPH$, resolving an open problem of Gharibian et al.\ \cite{gharibian2022quantum} and Agarwal et al.\ \cite{agarwal_et_al:LIPIcs.MFCS.2024.7}.

\begin{theorem}[Restatement of \cref{thm:amplification}]\label{thm:intro-amplification}
  $\pureQSigmai[r]\subseteq\QSigmai[7r](c',s')$ with $c'\ge 1-1/q(n)$ and $s'\le 1/q(n)$, where $q$ is an arbitrary polynomial.
\end{theorem}

That is, we show that any protocol in $\pureQSigmai[r]$ can be converted into one in $\QSigmai[7r]$ whose completeness (resp. soundness) is $1/\poly$-close to $1$ (resp. $0$) at the cost of a constant-factor increase in the number of alternations. 

Observe that \cref{thm:intro-amplification} says that every level of $\pureQPH$ is contained in some level of $\QPH$, i.e., that $\pureQPH \subseteq \QPH$. The reverse containment $\pureQPH \supseteq \QPH$ is easy to see: the provers can simply send purifications of the proofs in the $\QPH$ protocol.
Hence, the following corollary is immediate. 

\begin{corollary}[Restatement of \cref{cor:equal}]\label{cor:intro-equal}
   $\pureQPH = \QPH$. 
\end{corollary}

We emphasize that the equivalence of $\pureQPH$ and $\QPH$ is far from obvious. To illustrate, consider the following simple two-player game: Player 1 sends a state to the verifier, and Player 2, after learning Player 1’s message, must send the same state.\footnote{All of the protocols and proof systems we study can be viewed as games of \emph{perfect information}, meaning that each prover (or player) is fully aware of all moves made in the game so far. Indeed, even $\PH$ has this game-theoretic interpretation.}
The verifier runs a SWAP test, declaring ``Player 1 wins'' if the test fails and ``Player 2 wins'' if it passes. 
If the players are restricted to sending pure states, then Player 2 can always win with probability $1$, by perfectly replicating Player 1’s state. By contrast, if mixed states are allowed, Player 1 can send the maximally mixed state, in which case Player 2’s winning probability drops to approximately $1/2$.

Indeed, \cref{thm:intro-amplification,cor:intro-equal} are our most technically involved results. To establish them, we construct a new dimension-independent disentangler.

\begin{lemma}[Restatement of \cref{lem:pure-disentangler}]\label{lem:intro-pure-disentangler}
  Let $\calH = \CC^{d_1}\otimes\dotsm\otimes\CC^{d_s}$, $k\in \NN$, and $\delta>0$.
  There exist parameters $\ell\in\poly(\delta^{-1},k),m\in O(\delta^{-2})$ and a quantum channel 
  \[\Gamma\colon\calD(\calH^{\otimes 4\ell})\to\calD(\calH^{\otimes k}),\] 
  with the following properties:
  for all states $\rho_1,\rho_2,\rho_3,\rho_4\in\calD(\calH^{\otimes\ell})$, there exists a distribution $\{p_i\}_{i=1}^m$ over product states 
  \[
    \ket{\zeta_i} = \ket{\zeta_{i,1}} \otimes \dots \otimes \ket{\zeta_{i,s}}, \quad \ket{\zeta_{i,j}} \in \CC^{d_j}
    \]
  such that
  \begin{equation}
    \norm{\Gamma(\rho_1\otimes\rho_2\otimes\rho_3\otimes\rho_4) - \sum_{i=1}^m p_i \ketbra{\zeta_i}{\zeta_i}^{\otimes k}}_1 \le \delta.
  \end{equation}
  Furthermore, for 
 every pure product state $\ket{\psi}=\ket{\psi_1}\otimes\dotsm\otimes\ket{\psi_s}\in\calH$,
  $\Gamma(\ketbra{\psi}{\psi}^{\otimes4\ell}) = \ketbra{\psi}{\psi}^{\otimes k}$. 
\end{lemma}

Our construction builds on the disentangler of Jeronimo and Wu~\cite{JW24}, but strengthens it in a key way. While the Jeronimo-Wu channel guarantees closeness to a convex combination of product states, our disentangler ensures that the output is close to a convex combination of only $m = O(\delta^{-2})$ product states, where $\delta$ is a parameter that can be chosen. 
In other words, not only is the disentangled output structured, but it is also supported on a small set of product states, which is crucial for our amplification procedure. 
The tradeoff is that we use four unentangled input states whereas the Jeronimo-Wu channel only uses two.

Our final set of results concerns \emph{quantified Hamiltonian complexity}. We study natural generalizations of the local Hamiltonian problem, which asks: given a local Hamiltonian $H$, decide whether there exists a state $\ket{\psi}$ with energy $\braket{\psi|H|\psi} \leq a$, or if instead all states satisfy $\braket{\psi|H|\psi} \geq b$, promised one of these is the case. The local Hamiltonian problem is well known to be $\QMA$-complete~\cite{Kitaev2002,doi:10.1137/S0097539704445226}.

We extend this to the quantified setting, in analogy with quantified Boolean formulae \cite{arora2009computational}.  
For instance, the $\exists\forall$-mixed local Hamiltonian problem ($\exists\forall$-MLH) is defined as follows: given a local Hamiltonian $H$, decide whether there exists a mixed state $\rho$ such that for all mixed states $\sigma$, $\tr(H(\rho \otimes \sigma))\leq a$, or if instead, for all $\rho$, there exists a $\sigma$ such that $\tr(H(\rho\otimes\sigma))\ge b$, promised one of these is the case. 
For this problem, we obtain the following containment:

\begin{proposition}[Restatement of \cref{cor:npqma-conpqma}]\label{prop:intro-npqma}
The $\exists\forall$-mixed local Hamiltonian problem is in $\NP^\QMA \cap \coNP^\QMA$.
\end{proposition}

Since no complete problems are known for $\NP \cap \coNP$, we find it implausible that $\exists\forall$-MLH is complete for $\NP^\QMA \cap \coNP^\QMA$. Moreover, hardness for either $\NP^\QMA$ or $\coNP^\QMA$ would collapse these two classes, which also seems unlikely.

In addition to the mixed/local case, we also study the mixed/sparse, pure/local, and pure/sparse variants. These are defined analogously: the ``sparse'' condition means the Hamiltonian is row-sparse rather than local, while the ``pure'' versus ``mixed'' distinction specifies whether the quantified states are pure or mixed.
For the mixed/sparse and pure/local variants, our findings parallel the mixed/local case: we can establish containments but are unable to prove hardness. In fact, existing circuit-to-Hamiltonian constructions appear inadequate for obtaining hardness in these settings, suggesting that either new techniques would be required or that these variants fail to be complete problems for any class studied in this work.

The pure/sparse variant stands out as the only case where we can establish a completeness result. To capture this formally, we define the $\PSHSigmai[i]$ and $\PSHPii[i]$ problems (generalizing the $\QMAtwo$-complete separable sparse Hamiltonian problem \cite{chailloux2012complexity}), in which the input is a sparse Hamiltonian and the problem quantifies over $i$ quantum proofs (see \cref{def:PSH}). In this setting, we obtain the following completeness theorem:

\begin{theorem}[Restatement of \cref{thm:qph-completeness}]
$\PSHSigmai$ is $\pureQSigmai$-complete and $\PSHPii$ is $\pureQPii$-complete. 
\end{theorem}

These completeness results show that $\pureQPH$ admits natural complete problems, in contrast to the other variants of the quantum polynomial hierarchy. This highlights $\pureQPH$ as perhaps the most natural quantifier-based definition of the quantum polynomial hierarchy.

\subsection{Main Ideas}

\paragraph{Proving $\QMAtwo \subseteq \pureQSigmai \subseteq \QSigmai$}
The proof that $\QMAtwo \subseteq \pureQSigmai[2]$ is similar to the simple two-player game described earlier: one player sends a pure state, and the other must reproduce it exactly. With pure states, honesty can be enforced by a SWAP test. If the second player deviates, the SWAP test detects the inconsistency with constant probability.

A structural result of Harrow and Montanaro~\cite{HM13} ensures that in $\QMA(2)$ the two unentangled proofs may be taken to be identical. Thus, in $\pureQSigmai[2]$, the existential prover sends one copy of this proof $\ket{\psi}$, while the universal prover is challenged to send the same state. The verifier first applies a SWAP test and, if the test fails, immediately accepts (since the universal prover failed its task). If the SWAP test passes, the verifier then runs the $\QMA(2)$ verification on the two states.
Interestingly, the $\QMA(2)$ verification procedure is run on the \emph{post-measurement states} after the SWAP test is applied; a careful analysis shows that this simulation succeeds. 

The second inclusion, $\pureQSigmai[2] \subseteq \QSigmai[3]$, also relies on the SWAP test, but now it is used to enforce purity rather than equality. 
In the $\QSigmai[3]$ setting, the verifier receives three states: $\rho_1$ and $\rho_3$ from the existential prover, and $\rho_2$ from the universal prover. The goal is to simulate the $\pureQSigmai[2]$ protocol, where one prover supplies a pure state $\ket{\psi_1}$ and the other supplies a pure state $\ket{\psi_2}$.

A simple observation is that the universal prover in $\pureQSigmai[2]$ has no incentive to send a mixed state, since they move last; hence, we can safely take $\rho_2$ to play the role of $\ket{\psi_2}$. 
To certify that $\rho_1$ is effectively pure, the $\QSigmai[3]$ verifier asks for two copies, $\rho_1$ and $\rho_3$, and with some probability runs a SWAP test between them (rejecting if the test fails, and accepting otherwise). 
Otherwise, the verifier simulates the original $\pureQSigmai[2]$ verification using $\rho_1$ and $\rho_2$.

\paragraph{Amplification and $\pureQPH = \QPH$}
Our amplification procedure is the most technically involved part of this work. 
As a first step, we transform the standard alternating-proof system---where the provers take turns sending states in the order $\exists \forall \exists \dots$---into a system where, in each turn, a prover sends \emph{four} unentangled proofs simultaneously. We achieve this by increasing the number of rounds by a factor of $7$. In seven rounds the verifier receives states $\rho_1, \dots, \rho_7$, where the odd-indexed states ($\rho_1, \rho_3, \rho_5, \rho_7$) come from the existential prover and the even-indexed states ($\rho_2, \rho_4, \rho_6$) come from the universal prover. The verifier discards the even-indexed states by default, leaving a block of four unentangled states from the existential prover, $\rho_1 \otimes \rho_3 \otimes \rho_5 \otimes \rho_7$. An analogous construction handles the universal prover’s turns. In this way, each turn of the game is simulated by a block of seven rounds, giving us a protocol where the verifier receives four unentangled proofs per prover per turn.

Our goal now is to simulate $\pureQSigmai[r]$ in this $r$-round system where each prover sends four unentangled mixed proofs per turn (which, as explained, can be simulated in $\QSigmai[7r]$). For each block of four proofs, the verifier applies our disentangler (\cref{lem:intro-pure-disentangler}) to reduce the input to a convex mixture over a \emph{small set} of product states. Conceptually, in the $i$th round, each product state in the mixture can be viewed as $\ket{T}\ket{\psi}$, where $\ket{T}$ encodes a transcript of the first $i-1$ rounds and $\ket{\psi}$ is the candidate response for the current round given that transcript. In this way, every round can be interpreted as producing a distribution over transcript–answer pairs, and the verifier’s job is to make sure the prover sends a pair consistent with the ongoing interaction.

The verifier maintains a ``canonical transcript'' that grows round by round. At each step, the disentangler outputs a mixture of possible continuations, each consisting of a transcript prefix and a candidate answer. The verifier checks consistency between the candidate transcript and the actual transcript accumulated so far using repeated SWAP tests; if the tests succeed, the answer is appended to the canonical transcript. If they fail, the current prover loses the round (accepting if it is the universal prover’s turn, rejecting otherwise). After $r$ rounds, the verifier runs the original $\pureQSigmai[r]$ verifier $V_x$ on fresh copies of the canonical transcript and performs standard majority amplification to decide the outcome.

\paragraph{Quantified Hamiltonian Complexity}
Our approach to proving \cref{prop:intro-npqma} proceeds in two steps. First, we show that $\exists\forall$-MLH lies in $\NP^\QMA$ and, similarly, that $\forall\exists$-MLH lies in $\coNP^\QMA$. This step uses the fact that checking consistency of local density matrices---given local reduced density matrices, decide whether they arise from some global quantum state---can be solved with a single $\QMA$ query~\cite{Liu06}.
At a high level, the $\NP$ prover supplies classical descriptions of the reduced density matrices for the first witness $\rho$. A single $\QMA$ oracle call is then used to check that these matrices are indeed consistent with some global state. Once this is certified, the verifier can ``compress'' the Hamiltonian $H$ into a smaller Hamiltonian $H'$ that only acts on the Hilbert space corresponding to the universal prover’s state. Deciding whether all states $\sigma$ satisfy $\tr(H'\sigma) \ge b$ can then be handled with a $\coQMA$ query, which is equivalent to a $\QMA$ query.
The proof concludes by observing that $\exists\forall$-MLH and $\forall\exists$-MLH are equivalent by a minimax theorem.

We now turn to our completeness proofs of $\PSHSigmai$ and $\PSHPii$. These problems can be viewed as the $\exists\forall$- and $\forall\exists$-pure sparse Hamiltonian (PSH) problems generalized to an arbitrary constant number of quantifiers. 
The containment is relatively straightforward. We use the techniques of Aharonov and Ta-Shma~\cite{aharonov2007adiabatic} to efficiently simulate the dynamics of sparse Hamiltonians in $\BQP$, which places $\PSHSigmai$ and $\PSHPii$ inside $\pureQSigmai$ and $\pureQPii$, respectively.

The hardness direction requires extending the circuit-to-Hamiltonian framework to the quantified setting. Our construction generalizes the Hamiltonians of Chailloux and Sattath~\cite{chailloux2012complexity}, who proved that $\exists\exists$-PSH is $\QMA(2)$-complete.\footnote{In the terminology of Chailloux and Sattath~\cite{chailloux2012complexity}, the \emph{separable sparse Hamiltonian problem} is $\QMA(2)$-complete. In our language, this corresponds exactly to the $\exists\exists$-PSH problem.} We extend their approach to handle an arbitrary constant number of alternating quantifiers, ensuring that the Hamiltonian faithfully encodes the transcript of the underlying quantified proof system.

\section{Preliminaries}

For matrices, $\norm{\cdot}_1$ denotes the Schatten $1$-norm (also known as the trace norm or nuclear norm).
For quantum states $\rho, \sigma$, define the trace distance between $\rho$ and $\sigma$ as $\td(\rho, \sigma) \coloneqq \frac{1}{2} \norm{\rho - \sigma}_1$. If $\rho = \ketbra{\psi}{\psi}$ and $\sigma = \ketbra{\phi}{\phi}$, then $\td(\ket{\psi}, \ket{\phi}) = \sqrt{1 - \abs{\braket{\psi|\phi}}^2}$.
Let $\calD(\calH)$ denote the set of density operators on Hilbert space $\calH$.

The following is a basic fact about trace distance. 

\begin{fact}\label{fact:td}
Let $\rho$ and $\sigma$ be quantum states with $\td(\rho, \sigma) \leq \eps$. Then for any POVM element $0 \le M \le 1$,
$\abs{\tr(M\rho) - \tr(M\sigma)} \le \epsilon$.
\end{fact}

We also make use of two standard tools in quantum computation and quantum information, the SWAP test and the Gentle Measurement lemma, which we record below.

\begin{lemma}[SWAP test {\cite{buhrman2001quantum}}]\label{lemma:swap-test}
The SWAP test between two quantum states $\rho$ and $\sigma$ fails with probability $\frac{1}{2} - \frac{\tr(\rho \sigma)}{2}$.
\end{lemma}

\begin{lemma}[Gentle measurement {\cite{Win99}}]\label{lemma:gentle-measurement}
Consider a quantum state $\rho$ and a measurement operator $0 \leq M \leq 1$ where $\tr(\rho M) \geq 1 - \eps$. 
Then, for the post-measurement state
\begin{equation}
    \rho' \coloneqq \frac{\sqrt{M}\, \rho\, \sqrt{M}}{\tr(M \rho)},  
\end{equation}
we have 
\begin{equation}
   \td(\rho, \rho') \le 2\sqrt{\epsilon}.
\end{equation}
\end{lemma}

Finally, we turn to the formal definitions of $\QPH$ and $\pureQPH$. We begin by specifying the individual levels of these hierarchies.

\begin{definition}[$\QSigmai$]\label{def:QSigmai}
  Let $A=(\Ayes,\Ano)$ be a promise problem. We say that $A$ is in $\QSigmai(c, s)$ for poly-time computable functions $c, s: \NN \to [0, 1]$ if there exists a polynomial $p(n)$ and a poly-time uniform family of quantum circuits $\{V_x\}_{x \in \{0,1\}^*}$ such that for every $n$-bit input $x$, $V_x$ takes in quantum proofs $\rho_1,\dots,\rho_i\in\calD(\CC^{2^{p(n)}})$ and outputs a single qubit, {such that:}
  \begin{itemize}
    \item Completeness: $x\in \Ayes$ $\Rightarrow$ $\exists \rho_1 \forall \rho_2 \ldots Q_i \rho_i$ s.t. $\Pr[V_x \text{ accepts } \rho_1\otimes\dotsm\otimes\rho_i] \geq c(n)$.
    \item Soundness: $x\in \Ano$ $\Rightarrow$ $\forall \rho_1 \exists \rho_2 \ldots \overline{Q}_i \rho_i$ s.t. $\Pr[V_x \text{ accepts }  \rho_1\otimes\dotsm\otimes\rho_i] \leq s(n)$.
  \end{itemize}
  Here, $Q_i$ is $\exists$ when $i$ is odd and $\forall$ otherwise, and $\overline{Q}_i$ is the complementary quantifier to $Q_i$. Finally, define
  \begin{equation}
  \QSigmai := \bigcup_{{ c(n) - s(n) \in \Omega(1/\poly(n))}} \QSigmai(c, s).
  \end{equation}
  Define $\pureQSigmai$ analogously, restricting $\rho_1,\dots,\rho_i$ to pure states.
\end{definition}

\begin{remark}\label{remark:padding}
  All messages being $p(n)$ qubits is without loss of generality, even for $\pureQSigmai$, as the verifier can project messages onto a smaller subspace and let the sender lose if the projection fails.
\end{remark}

$\QPH$ and $\pureQPH$ are the union over all levels of their hierarchies.

\begin{definition}[$\QPH$ and $\pureQPH$ {\cite{gharibian2022quantum,agarwal_et_al:LIPIcs.MFCS.2024.7}}]
The quantum polynomial hierarchy is defined as 
\[
\QPH \coloneqq \bigcup_{i=0}^\infty \QSigmai, 
\]
and the pure quantum polynomial hierarchy is defined as 
\[
\pureQPH \coloneqq \bigcup_{i=0}^\infty \pureQSigmai. 
\]
\end{definition}

One has to be careful when discussing oracles to promise problems, e.g., $\NP^\QMA$.
We say a deterministic Turing machine $M$ with access to a promise oracle $O = (\Oyes,\Ono)$ accepts/rejects \emph{robustly} if $M$ accepts/rejects regardless of how invalid queries are answered (see also \cite[Definition 3]{Goldreich2006}).

\begin{definition}[$\NP$ with promise oracle {\cite[Footnote 3]{agarwal_et_al:LIPIcs.MFCS.2024.7}}]\label{def:promise-oracle}
    Let $O$ be a promise problem.
    We say $A\in \NP^O$ if there exists a polynomial-time deterministic Turing machine $M$, such that
    \begin{itemize}
        \item $x\in\Ayes\ \Rightarrow\ \exists y\colon M^O(x,y)$ accepts robustly.
        \item $x\in\Ano\ \Rightarrow\ \forall y\colon M^O(x,y)$ rejects robustly.
    \end{itemize}
\end{definition}

This definition may be considered the weakest ``reasonable'' definition for $\NP$ with promise oracle, without outright forbidding invalid queries.

\input{qma2}

\input{amplification}

\input{np_qph}

\section{Open Problems}

We conclude by highlighting several natural directions for future work.  

\begin{enumerate}
    \item Are there complete problems for variants of the quantum polynomial hierarchy beyond $\pureQPH$? For example, can the pure/local, mixed/local, or mixed/sparse Hamiltonian problems be shown complete for any class?  
    
    \item It is striking that the local variants of the quantified Hamiltonian problem only required quantum computation in the \emph{oracle} part of the machine. For instance, in \cref{prop:puresuperqma}, the base machine is merely $\NP$. A natural direction is to better understand the relationship between $\NP^\QMA$ and $\QMA^\QMA$.  
    
    \item Is $\exists\forall$-$k$-PLH complete for $\NP^{\Vert \pureSuperQMA[2]}$? A positive answer would imply that \[\NP^{\Vert \pureSuperQMA[2]} \subseteq \pureQSigmai[2],\] giving the first connection between oracle-based and quantifier-based definitions of the quantum polynomial hierarchy.  
    
    \item  Can one construct disentanglers with guarantees stronger than those in \cref{lem:disentangler}? In particular, is it possible to design a disentangler whose output is always (close to) a pure state?  
\end{enumerate}

\section*{Acknowledgements}
We thank Justin Yirka, Sevag Gharibian, and William Kretschmer for helpful conversations. 
SG was supported in part by an IBM PhD Fellowship.
DR was supported in part by the DFG under grant number 432788384.
This work was done in part while the authors were visiting the Simons Institute for the Theory of Computing.

\bibliographystyle{alphaurl}
\bibliography{bibliography.bib}

\appendix

\end{document}

%% file: qma2.tex
\section{Sandwiching the Second Level of \texorpdfstring{$\pureQPH$}{pureQPH}}

We prove the inclusions $\QMAtwo \subseteq \pureQSigmai[2] \subseteq \QSigmai[3]$, i.e., that the second level of $\pureQPH$ lies between $\QMAtwo$ and the third level of $\QPH$. 
It is known from prior work that $\QSigmai[3] \subseteq \NEXP$ \cite{gharibian2022quantum}. 
We begin by establishing the first inclusion, showing that any $\QMA(2)$ protocol can be simulated within the second level of $\pureQPH$.

\begin{theorem}\label{thm:qma2-in-psigma2}
  $\QMA(2)\subseteq\pureQSigmai[2]$.
\end{theorem}
\begin{proof}
  Let $A\in \QMA(2)$.
  By \cite{HM13}, there exists a verifier $V$, such that
\begin{subequations}
\begin{align}
  \forall x \in \Ayes,\, \exists \ket{\psi} \colon\;
    \tr\!\left( H_x \big( \ketbra{\psi}{\psi}_A \otimes \ketbra{\psi}{\psi}_B \big) \right)
    &\ge 1 - \epsilon, \label{eq:qma2-yes} \\
  \forall x \in \Ano,\, \forall \rho \,\forall \sigma \colon\;
    \tr\!\left( H_x ( \rho_A \otimes \sigma_B ) \right)
    &\le \epsilon, \label{eq:qma2-no}
\end{align}
\end{subequations}
  where $H_x$ denotes the POVM element corresponding to acceptance on input $x$, and $\epsilon\in 2^{-O(n)}$.
  We now construct a $\pureQSigmai[2]$ verifier $V'$ as follows.
  On input $x$, $V'$ receives two states $\ket{\psi}_A$ and $\ket{\phi}_B$ and acts as follows:
  \begin{enumerate}[label=(\arabic*)]
    \item Perform a SWAP test on registers $A$ and $B$. If the SWAP test fails, then accept.
    \item Otherwise, run $V$ on registers $A$ and $B$, and accept only if $V$ accepts.
  \end{enumerate}
  We let $H_x'$ denote the POVM element corresponding to acceptance on input $x$ for the verifier $V'$. 

  First, we prove the soundness of $V'$, which is straightforward. 
  Suppose $x\in\Ano$. By \cref{eq:qma2-no}, we have 
  \begin{equation}
    \forall \ket{\psi}\colon \tr\left(H_x(\ketbra{\psi}{\psi}\otimes \ketbra{\psi}{\psi})\right)\le\epsilon.
  \end{equation}
  In this case, the no-prover will always send $\ket{\phi}_B = \ket{\psi}_A$, which ensures the SWAP test accepts with probability $1$ (\cref{lemma:swap-test}) and leaves the state undisturbed.
  The verifier then proceeds to run $V$ on $\ket{\psi}_A \otimes \ket{\psi}_A$, which will accept with probability at most $\eps$. Therefore, the overall acceptance probability of $V'$ is at most $\eps$.
  
  Now suppose $x\in\Ayes$. We will show
  \begin{equation}
  \exists\ket{\psi}\,\forall\ket{\phi}\colon\tr\left(H'_x(\ketbra{\psi}{\psi}\otimes\ketbra{\phi}{\phi})\right)\ge c
  \end{equation}
  for some constant $c > 0$, which suffices to complete the proof.
  Let $\ket{\psi}$ be the witness guaranteed by \cref{eq:qma2-yes}.
  For an arbitrary state $\ket{\phi}$, define $\delta \coloneqq \td(\ket{\psi}, \ket{\phi})$, so $\abs{\braket{\psi|\phi}}^2 = 1-\delta^2$.
  By \cref{lemma:swap-test}, $V'$ accepts in step (1) (i.e., the SWAP test fails) with probability
  \begin{equation}
    \frac12-\frac12\abs{\braket{\psi|\phi}}^2 = \frac{\delta^2}2.
  \end{equation}
 If the SWAP test succeeds, let $\rho_{AB} = \ketbra{\psi}{\psi}\otimes\ketbra{\phi}{\phi}$, and let $\rho'_{AB}$ be the post-measurement state conditioned on the SWAP test succeeding. 
By the Gentle Measurement Lemma (\cref{lemma:gentle-measurement}), 
\begin{equation}
    \td(\rho,\rho') \leq \sqrt{2} \delta.
\end{equation}
Let $\rho^* = \ketbra{\psi}{\psi}\otimes \ketbra{\psi}{\psi}$. By the triangle inequality, 
  \begin{equation}
    \td(\rho^*, \rho') \le \td(\rho^*, \rho) + \td(\rho, \rho') \le \delta + \sqrt{2}\delta = (1 + \sqrt{2}) \delta. 
  \end{equation}
  Thus, by \cref{fact:td} and \cref{eq:qma2-yes}, $V$ accepts $\rho'$ with probability at least
  \begin{equation}
    \tr(H_x\rho') \ge 1-\epsilon - \td(\rho^*,\rho') \ge 1-\epsilon- (1+\sqrt{2})\delta. 
  \end{equation}
  Therefore, the overall acceptance probability of $V'$ on $\rho_{AB}$ is
  \begin{equation}
    \tr(H_x'\rho) = \frac{\delta^2}2 + \left(1-\frac{\delta^2}2\right)\max\left\{0,1-\epsilon-(1+\sqrt{2})\delta\right\}. 
  \end{equation}
Consider the case where $\delta \geq \frac{1-\eps}{1 + \sqrt{2}}$. The $\max$ term vanishes, so the acceptance probability is simply $\frac{\delta^2}{2}$, minimized when $\delta$ is as small as possible, i.e., $\delta = \frac{1-\eps}{1 + \sqrt{2}}$.
Now consider the case where $\delta \leq \frac{1-\eps}{1 + \sqrt{2}}$. 
The expression becomes $1 - \eps - (1 + \sqrt 2) \delta + \frac{\eps}{2} \delta^2 + \frac{1 + \sqrt{2}}{2} \delta^3$, which is decreasing on the interval $[0, \frac{1-\eps}{1 + \sqrt{2}}]$, so the minimum also occurs at $\delta = \frac{1-\eps}{1 + \sqrt{2}}$. 
In both cases, the minimizing value is $\delta^\star = \frac{1-\eps}{1 + \sqrt{2}}$, giving 
\begin{equation}
    \tr(H_x'\rho) = \frac{(1-\eps)^2}{2(1+\sqrt{2})^2} = \left( \frac{3}{2} - \sqrt{2}\right)\left( 1 - \eps\right)^2 > 0.085, 
\end{equation}
for sufficiently large $n$. 
\end{proof}

We now show that the second level of $\pureQPH$ is contained in the third level of $\QPH$.

\begin{theorem}\label{thm:pqsigma2-in-qsigma3}
  $\pureQSigmai[2] \subseteq\QSigmai[3]$.
\end{theorem}
\begin{proof}
  Let $A\in\pureQSigmai[2]$.
  Then there exists a verifier $V$ with completeness $c$, soundness $s$, and $c-s\ge n^{-O(1)}$, together with a POVM element $H_x$ for each input $x$, such that
  \begin{subequations}
    \begin{align}
      \forall x\in \Ayes,\,  \exists \ket\psi\;\forall\ket\phi\colon\tr(H_x(\ketbra{\psi}{\psi}_A\otimes \ketbra{\phi}{\phi}_B))&\ge c,\label{eq:pQSi2-yes}\\
      \forall x\in \Ano,\, \forall\ket{\psi}\;\exists\ket{\phi}\colon\tr(H_x(\ketbra{\psi}{\psi}_A\otimes \ketbra{\phi}{\phi}_B)) &\le s.\label{eq:pQSi2-no}
    \end{align}
  \end{subequations}
  Define a $\QSigmai[3]$ verifier $V'$ with POVM $H_x'$ on input $x$ as follows. Let $\rho_1$ be the proof sent by the yes-prover in the first round, $\rho_2$ be the proof sent by the no-prover in the second round, and $\rho_3$ be the proof sent by the yes-prover in the third round. $V'$ then proceeds as follows:
  \begin{enumerate}[label=(\arabic*)]
    \item With probability $p$, run $V(\rho_1,\rho_2)$, accepting or rejecting according to $V$'s output.
    \item With probability $1-p$, run a SWAP test between $\rho_1$ and $\rho_3$, and accept only if the SWAP test passes.
  \end{enumerate}
  We must exhibit completeness/soundness paramters $c',s'$ with $c'-s'\ge n^{-O(1)}$, such that 
  \begin{subequations}
    \begin{align}
      \forall x\in \Ayes,\, \exists \rho_1\;\forall\rho_2\;\exists\rho_3\colon\tr(H'_x(\rho_1\otimes\rho_2\otimes\rho_3))&\ge c'\label{eq:QSi3-yes}\\
      \forall x\in \Ano,\ \forall \rho_1\;\exists\rho_2\;\forall\rho_3\colon\tr(H'_x(\rho_1\otimes\rho_2\otimes\rho_3))&\le s'.\label{eq:QSi3-no}
    \end{align}
  \end{subequations}
  
  Suppose $x \in \Ayes$. We have
  \begin{equation}
    c'=(1-p) + p\cdot c = 1 - p(1-c),
  \end{equation}
  because the yes-prover can send $\rho_1=\rho_3= \ketbra{\psi}{\psi}$ from \cref{eq:pQSi2-yes} and the no-prover gains no advantage from sending a mixed state by convexity.

  Now suppose $x\in\Ano$. Let $1-\frac{\delta}{2}$ be the probability that the SWAP test between $\rho_1$ and $\rho_3$ passes. 
  Then, by \cref{lemma:swap-test}, $\tr(\rho_1 \rho_3) = 1 - \delta$, which implies $\lmax(\rho_1)\ge 1-\delta$ by H{\"o}lder's inequality. Let $\ket{\psi}$ be the corresponding eigenvector. 
  By \cref{eq:pQSi2-no}, there exists a ``refutation'' $\ket{\phi}$ (depending on $\ket{\psi}$) such that 
  \begin{equation}
      \tr(H_x(\ketbra{\psi}{\psi}_A\otimes \ketbra{\phi}{\phi}_B)) \le s.
  \end{equation}
  Let the no-prover  send $\rho_2 = \ketbra{\phi}{\phi}$. 
  Using linearity and that $0 \leq H_x \leq 1$, we get 
  \begin{align}
     \tr\left(H_x (\rho_1 \otimes \rho_2)\right)
     &= \lmax \tr\left(H_x (\ketbra{\psi}{\psi} \otimes \rho_2)\right) + (1-\lmax)\tr\left(H_x (\sigma \otimes \rho_2)\right) \\
     &\leq \lmax \cdot s + (1-\lmax) \cdot 1 \\
     &= s + (1-\lmax)(1-s) \\
     &\leq (1 - \delta) s + \delta,  
  \end{align}
  where, in the first line we use the fact that $\rho_1 = \lmax \ketbra{\psi}{\psi} + (1-\lmax) \sigma$ for some state $\sigma$, and, in the last line, we use the fact that $\lmax \geq 1 - \delta$.
  The overall acceptance probability $s'$ of $V'$ is thus
  \begin{align}
    s' 
    &= \max_{\delta\in[0,1]}\left((1-p)(1-\frac{\delta}{2}) + p((1-\delta)s + \delta)\right) \\ 
    &= 1 - p + ps + \delta\left(p (1-s) - \frac{1-p}{2} \right) \\ 
    &= \max\left(1-p+ps, \frac{1+p}{2}\right).
  \end{align}
We choose $p = \frac{1}{3-2s}$, which ensures that $p\in(0, 1]$ because $s\in[0,1]$. 
Then
\begin{equation}
   1-p+ps = 1 - \frac{1}{3-2s} + \frac{s}{3-2s} = \frac{3-2s-1+s}{3-2s} = \frac{2-s}{3-2s}, 
\end{equation}
   and
\begin{equation}
   \frac{1+p}{2} = \frac{1 + \frac{1}{3-2s}}{2} = \frac{\frac{3-2s+1}{3-2s}}{2} = \frac{4-2s}{2(3-2s)} = \frac{2-s}{3-2s}, 
\end{equation}
so the two terms in the $\max$ function become equal.
Therefore, $s' = \frac{2-s}{3-2s}$.
Our completeness parameter $c'$ becomes 
\begin{equation}
   c' = 1 - p(1-c) = 1 - \frac{1-c}{3-2s} = \frac{3-2s - (1-c)}{3-2s} = \frac{2-2s+c}{3-2s}.
\end{equation}
Therefore, the completeness/soundness gap is
   \begin{equation}
   c' - s' = \frac{2-2s+c}{3-2s} - \frac{2-s}{3-2s} = \frac{c-s}{3-2s}.
   \end{equation}
   Because $c-s \ge n^{-O(1)}$ by assumption and $3-2s \leq 3$, we have $c'-s' \ge \frac{c-s}{3} \ge n^{-O(1)}$, which completes the proof.
\end{proof}

%% file: amplification.tex
\section{Amplification of \texorpdfstring{$\pureQPH$}{pureQPH} via Disentanglers}

In this section, we give an error reduction result for $\pureQPH$. 
We show that any protocol in $\pureQSigmai[r]$ can be converted into one in $\QSigmai[7r]$ whose completeness is arbitrarily close to 1 and whose soundness is arbitrarily close to 0 (up to $1/\poly(n)$). 
In other words, we can amplify the gap between YES and NO cases at the cost of a constant-factor increase in the number of alternations.

\begin{restatable}{theorem}{amplification}\label{thm:amplification}
  $\pureQSigmai[r]\subseteq\QSigmai[7r](c',s')$ with $c'\ge 1-1/q(n)$ and $s'\le 1/q(n)$, where $q$ is an arbitrary polynomial.
\end{restatable}

\cref{thm:amplification} shows that every level of $\pureQPH$ is contained in $\QPH$. 
The reverse inclusion $\QPH\subseteq\pureQPH$ is straightforward: given mixed-state proofs, the provers can instead send purifications, and the verifier can trace out the auxiliary registers to recover the original mixed states. Putting these two directions together, we conclude that the hierarchies are equal. 

\begin{corollary}\label{cor:equal}
  $\pureQPH = \QPH$.
\end{corollary}

The remainder of this section is devoted to proving \cref{thm:amplification}. Before presenting the proof, we construct a disentangler tailored to our purposes, building on the dimension-independent disentangler of Jeronimo and Wu~\cite{JW24}.

\begin{theorem}[Disentangler \cite{JW24}]\label{lem:disentangler}
  Let $d,\ell\ge k\in\NN$ and $\calH = \CC^d$.
  There exists an efficient quantum channel $\Lambda:\calD(\calH^{\otimes 2\ell})\to\calD(\calH^{\otimes k})$, such that for all states $\rho_1,\rho_2$, there exists a distribution $\mu$ on pure states $\ket{\psi}\in\CC^d$, such that
  \begin{equation}\label{eq:disentangler-error}
    \norm{\Lambda(\rho_1\otimes\rho_2) - \int{\psi}^{\otimes k}d\mu}_1 \le \tilde O\left(\left(\frac{k^3}{\ell}\right)^{1/4}\right).
  \end{equation}
  Furthermore, $\Lambda({\psi}^{\otimes2\ell})={\psi}^{\otimes k}$ for all $\ket{\psi}\in\calH$.
\end{theorem}

Note that by Carathéodory's theorem \cite{Carathéodory1911}, we always have 
\begin{equation}\label{eq:caratheodory}
  \int{\psi}^{\otimes k}d\mu = \sum_{i} p_i{\psi_i}^{\otimes k}
\end{equation}
for some distribution $p_i$ on pure states $\ket{\psi_i}$.

We also use the following facts about the SWAP test and the product test of Harrow and Montanaro \cite{HM13}.
We let $\Pswap(\rho,\sigma)$ and $\Pprod(\rho, \sigma)$ denote the acceptance probability of the SWAP test and the product test, respectively. We let $\Pprod(\rho)\coloneq\Pprod(\rho,\rho)$. 

\begin{lemma}[Product Test {\cite[Theorem 3]{HM13}}]\label{lem:product-test}
  Given $\ket{\psi}\in\CC^{d_1}\otimes\dotsm\otimes\CC^{d_n}$, let 
  \begin{equation}
    1-\epsilon = \max\bigl\{\abs{\braket{\psi|\phi_1,\dots,\phi_n}}^2 \bigm| \ket{\phi_i}\in\CC^{d_i}, i\in[n] \bigr\}.
  \end{equation}
  Then $\Pprod({\psi}) = 1 - \Theta(\epsilon)$. 
\end{lemma}

\begin{lemma}[{\cite[Proof of Lemma 5]{HM13}}]\label{lem:Pprod-bound}
  $\Pprod(\rho,\sigma) \le \frac12(\Pprod(\rho) + \Pprod(\sigma))$.
\end{lemma}

\begin{lemma}\label{lem:prod-swap}
  $\Pprod(\rho,\sigma)\le \Pswap(\rho,\sigma)$. 
\end{lemma}
\begin{proof}
  Denote the registers of $\rho$ by $\calA_1,\dots,\calA_n$ and $\sigma$ by $\calB_{1},\dots,\calB_{n}$.
  For each $i$, let $F_{i}$ denote the swap operator between $\calA_i$ and $\calB_i$.
  Note that $F_i^2=I$, so the eigenvalues of $F_i$ are $\pm1$. 
  The projectors onto the accepting subspaces are given by
  \begin{equation}
    \Piswap = \frac{I+F}{2},\quad F = \prod_{i=1}^n F_i,\qquad \Piprod = \prod_{i=1}^n\frac{I+F_i}{2}.
  \end{equation}
  We claim that $\Piswap \succeq \Piprod$.
  Since the $F_i$ act on disjoint registers, they commute. 
  Thus there exists a common eigenbasis consisting of product vectors $\ket\psi=\bigotimes_{i=1}^n\ket{\psi_i}_{\calA_i\calB_i}$, where each $\ket{\psi_i}$ is an eigenvector of $F_i$ with eigenvalue $\lambda_i\in\{1,-1\}$.
  For such an eigenstate $\ket\psi$, we have $\bra{\psi}\Piprod\ket{\psi} = 1$ iff all $\lambda_i=1$.
  On the other hand, $\bra{\psi}\Piswap\ket{\psi} = 1$ iff there is an even number of negative $\lambda_i$.
  Thus, whenever $\Piprod$ accepts, $\Piswap$ also accepts, while the converse need not hold. This establishes $\Piswap \succeq \Piprod$, and hence $\Pprod(\rho,\sigma) \leq \Pswap(\rho,\sigma)$.
\end{proof}

To proceed, we need the following combinatorial lemma about distributions. Intuitively, it says that if a certain event happens with non-negligible probability, then a small ``hitting set'' of outcomes suffices to capture the event with high conditional probability.

\begin{lemma}\label{lem:peaked}
  Let $p = (p_i), q = (q_j)$ be independent probability distributions over $[N]$. Let $S\subseteq [N]^2$ satisfy $\Pr[(i,j)\in S]=\epsilon$.
  Then there exists a set $X\subseteq [N]$ of size 
  \[m\le \left\lceil \frac{1}{e\epsilon\gamma}\right\rceil,\]
  such that $\Pr[\exists k\in X:(i,k)\in S\mid (i,j)\in S]\ge 1-\gamma$.
\end{lemma}
\begin{proof}
  For each $i\in[N]$, define $S_i = \{k\in [N]\mid (i,k)\in S\}$, and $q(S_i) = \Pr_{k\sim q}[k\in S_i] = \sum_{k\in S_i} q_k$.
  Then $\epsilon = \sum_{i=1}^N p_i q(S_i)$.

  Let $E$ be the event that $(i,j)\in S$, and $F_X$ the event that $S_i\cap X =\emptyset$.
  Our goal is to find a set $X$ of size $m$ such that
  \begin{equation}
    \Pr[F_X\mid E] = \frac{\Pr[F_X\cap E]}{\Pr[E]} \le \gamma,
  \end{equation}
  which is equivalent to finding a set $X$ such that  
\begin{equation}\label{eq:PrFXE}
\Pr[F_X \cap E] = \sum_{i : S_i \cap X = \emptyset} p_i q(S_i) \leq \gamma \epsilon.
\end{equation}
  
  Now construct $X$ by sampling $m$ i.i.d. elements $x_1,\dots,x_m\sim q$ and let $X=\{x_1,\dots,x_m\}$. 
  Define the random variable $Z = \sum_{i:S_i\cap X = \emptyset} p_i q(S_i)=\Pr[F_X\cap E]$.
  Then
  \begin{subequations}
  \begin{align}
    \E_X[Z] &= \E_X\left[\sum_{i=1}^N p_i q(S_i)\cdot\indic(S_i\cap X = \emptyset)\right]  \\
    &= \sum_{i=1}^N p_iq(S_i)\E_X[\indic(S_i\cap X = \emptyset)]\\
    &=\sum_{i=1}^N p_i q(S_i) \Pr_X[S_i\cap X=\emptyset] \\ 
    &=\sum_{i=1}^N p_i q(S_i) (1-q(S_i))^m\\
    &\le \sum_{i=1}^N p_i q(S_i)e^{-m q(S_i)}\label{eq:EZ:c}\\
    &\le \frac{1}{em}\sum_{i=1}^N p_i \label{eq:EZ:d}\\ 
    &= \frac{1}{em},
  \end{align}
  \end{subequations}
  where \cref{eq:EZ:c} uses $1-x\le e^{-x}$ and \cref{eq:EZ:d} uses that $f(x)=xe^{-mx}$ has a global maximum at $x=1/m$.
  For $m = \lceil1/(e\epsilon\gamma)\rceil$, we have $\E_X[Z] \le \gamma\epsilon$ as desired. Thus, there must exist an $X$, for which \cref{eq:PrFXE} holds, as desired.
\end{proof}

We now combine our combinatorial lemma with the dimension-independent disentangler of Jeronimo and Wu to obtain a new disentangler suited for our setting.

\begin{lemma}\label{lem:pure-disentangler}
  Let $\calH = \CC^{d_1}\otimes\dotsm\otimes\CC^{d_s}$, $k\in \NN$, and $\delta>0$.
  There exist parameters $\ell\in\poly(\delta^{-1},k),m\in O(\delta^{-2})$ and a quantum channel 
  \[\Gamma\colon\calD(\calH^{\otimes 4\ell})\to\calD(\calH^{\otimes k}),\] 
  with the following properties:
  for all states $\rho_1,\dots,\rho_4\in\calD(\calH^{\otimes\ell})$, there exists a distribution $\{p_i\}_{i=1}^m$ over product states 
  \[
    \ket{\zeta_i} = \ket{\zeta_{i,1}} \otimes \dots \otimes \ket{\zeta_{i,s}}, \quad \ket{\zeta_{i,j}} \in \CC^{d_j}
    \]
  such that
  \begin{equation}
    \norm{\Gamma(\rho_1\otimes\rho_2\otimes\rho_3\otimes\rho_4) - \sum_{i=1}^m p_i \ketbra{\zeta_i}{\zeta_i}^{\otimes k}}_1 \le \delta.
  \end{equation}
  Furthermore, for 
 every pure product state $\ket{\psi}=\ket{\psi_1}\otimes\dotsm\otimes\ket{\psi_s}\in\calH$,
  $\Gamma(\ketbra{\psi}{\psi}^{\otimes4\ell}) = \ketbra{\psi}{\psi}^{\otimes k}$. 
\end{lemma}
\begin{proof}
Let $\Lambda \colon \calD(\calH^{\otimes 2\ell}) \to \calD(\calH^{\otimes (k+k')})$ denote the channel from \cref{lem:disentangler}, now parameterized so that the output has size $k+k'$ rather than $k$.  
We define $\Gamma(\rho_1 \otimes \rho_2 \otimes \rho_3 \otimes \rho_4)$ as follows:
  \begin{enumerate}
    \item Apply $\Lambda$ twice to obtain $\sigma_1 = \Lambda(\rho_1\otimes\rho_2)$ and $\sigma_2 = \Lambda(\rho_3\otimes\rho_4)$ on registers $\calA_1,\dots,\calA_{k+k'}$ and $\calB_1,\dots,\calB_{k+k'}$ respectively.
    \item For $i=1,\dots,k'$, perform a product test between $\calA_{i}$ and $\calB_{i}$.
    \begin{enumerate}
      \item If all product tests accept, output registers $\calA_{k'+1},\dots,\calA_{k+k'}$.
      \item Otherwise, output $\ket{\mathbf0}^{\otimes k}$, where $\ket{\mathbf0} = \ket{0_1,\dots,0_s}\in\calH$.
    \end{enumerate}
  \end{enumerate}
  Let $\eta = \Gamma(\rho_{1}\otimes\rho_2\otimes\rho_3\otimes\rho_4)$ be the output of our channel.
  Note that, by \cref{lem:disentangler,eq:caratheodory}, $\sigma_1$ and $\sigma_2$ in step 1 can be approximated as 
   \begin{equation}\label{eq:gamma2}
    \norm{\sigma_1\otimes\sigma_2 - \sigma'_1\otimes\sigma_2'}_1\le2\epsilon_\Lambda,\qquad  \sigma_1' =\sum_{i=1}^M p_i{\ketbra{\psi_i}{\psi_i}}^{\otimes (k+k')},\sigma_2' = \sum_{j=1}^M q_j{\ketbra{\phi_j}{\phi_j}}^{\otimes(k+k')},
  \end{equation} 
  where $\epsilon_\Lambda$ denotes the error due to $\Lambda$ (\cref{eq:disentangler-error}). We will eventually choose $\ell$ so that $2 \epsilon_\Lambda \le \tfrac{\delta}{2}$.
  
  Let $\Gamma_2$ be the channel corresponding to step 2 in the definition of $\Gamma$ above.
  By contractivity, we have
  \begin{equation}
  \norm{\eta - \Gamma_2(\sigma_1'\otimes\sigma_2')}_1 = 
  \norm{\Gamma_2(\sigma_1 \otimes \sigma_2) - \Gamma_2(\sigma_1'\otimes\sigma_2')}_1
  \le 2\epsilon_{\Lambda}.
  \end{equation}
  Let $\pacc$ be the probability that all product tests accept in step (2a).
  If $\pacc \le \delta/4$, then 
  \begin{equation}
  \norm{\eta - \ketbra{\mathbf0}{\mathbf0}^{\otimes k}}_1 \le 2\epsilon_\Lambda + 2\pacc \le \delta.
  \end{equation} 
  Hence, assume $\pacc>\delta/4$.
  It holds that
  \begin{equation}
    \pacc = \sum_{i,j}p_iq_j \Pprod({\ketbra{\psi_i}{\psi_i}},{\ketbra{\phi_j}{\phi_j}})^{k'} \eqcolon \sum_{i,j}p_iq_j c_{ij}.
  \end{equation}
  Define $\epsilon_S \coloneqq \alpha\pacc$ for $\alpha\in(0,1)$ to be determined later. Let $S = \{(i,j) \mid c_{ij}\ge \epsilon_S \}$.
  Then
  \begin{equation}
    \pacc = \E[c_{ij}] \le \epsilon_S\Pr[S^c] + \Pr[S] \;\Longrightarrow\; \Pr[S] \ge (1-\alpha)\pacc.
  \end{equation}
  We will use \cref{lem:peaked} to approximate $\ket{\psi_i}$ with $(i,j)\in S$ with a small distribution.
  However, there is still a chance that the product tests in (2a) accept (event $A$) even if $(i,j)\notin S$:
  \begin{equation}\label{eq:Sc}
    \Pr[A\cap S^c] = \sum_{ij\in S^c} p_iq_jc_{ij} \le \sum_{ij}p_iq_j\epsilon_S = \alpha\pacc
  \end{equation}
  For all $(i,j)\in S$,
  \begin{equation}
    \Pprod(\ketbra{\psi_i}{\psi_i}, \ketbra{\phi_j}{\phi_j}) = c_{ij}^{1/k'} \ge (\epsilon_S)^{1/k'} \eqcolon \tau,
  \end{equation}
  where $k'\ge-t\ln(\alpha\delta/4) > -t\ln(\epsilon_S)$ gives $\tau=(\epsilon_S)^{1/k'}\ge e^{-1/t}\ge1-1/t$ for $t>1$ to be determined later.
  By \cref{lem:peaked} with parameters $\epsilon\gets\Pr[S]$ and $\gamma$ to be determined later, there exists a set $X\subseteq[M]$ of size $m=O(1/((1-\alpha)\pacc\gamma))\le O(1/(\delta\gamma))$ (recall $\pacc > \delta/4$) such that
  \begin{equation}\label{eq:S'}
    \Pr[i \in S'\mid (i,j) \in S]\ge 1-\gamma,\qquad S' = \{i\in[M]\mid \exists j_i\in X\colon (i,j_i)\in S\}.
  \end{equation}
  For all $i\in S'$, define $\ket{\eta_i} = \ket{\phi_{j_i}}$ for some $j_i$ with $(i,j_i)\in S$.
  These $\ket{\eta_i}$ must be close to product as $\Pprod(\ketbra{\psi_i}{\psi_i},\ketbra{\eta_i}{\eta_i})\ge\tau\ge 1-1/t$. 
  By \cref{lem:Pprod-bound},
  \begin{equation}
    1-1/t\le \Pprod(\ketbra{\psi_i}{\psi_i},\ketbra{\eta_i}{\eta_i}) \le \frac12\bigl(\Pprod(\ketbra{\psi_i}{\psi_i})+\Pprod(\ketbra{\eta_i}{\eta_i})\bigr)\;\Longrightarrow\;\Pprod(\ketbra{\eta_i}{\eta_i})\ge 1-2/t.
  \end{equation}
  By \cref{lem:product-test}, there exists $\ket{\zeta_i} = \ket{\zeta_{i,1}}\otimes\dotsm\otimes\ket{\zeta_{i,s}}\in \calH$ such that $\abs{\braket{\eta_i|\zeta_i}}^2\ge 1-O(1/t)$.
  Additionally, $\abs{\braket{\psi_i|\eta_i}}^2 = 2\Pswap(\ketbra{\eta_i}{\eta_i},\ketbra{\psi_i}{\psi_i})-1\ge2\Pprod(\ketbra{\eta_i}{\eta_i},\ketbra{\psi_i}{\psi_i})-1\ge1-2/t$.
  Hence, for an appropriate constant $C$,
  \begin{equation}\label{eq:psik}
    \begin{aligned}
    \norm{\ketbra{\psi_i}{\psi_i}^{\otimes k} - \ketbra{\zeta_i}{\zeta_i}^{\otimes k}}_1 
    &\le \norm{\ketbra{\psi_i}{\psi_i}^{\otimes k} - \ketbra{\eta_i}{\eta_i}^{\otimes k}}_1 +  \norm{\ketbra{\eta_i}{\eta_i}^{\otimes k} - \ketbra{\zeta_i}{\zeta_i}^{\otimes k}}_1 \\
    &= 2\sqrt{1-\abs{\braket{\psi_i|\eta_i}}^{2k} } + 2\sqrt{1-\abs{\braket{\eta_i|\zeta_i}}^{2k}}\\
    &\le \sqrt{Ck/t},
    \end{aligned}
  \end{equation}
  where the last step uses Bernoulli's inequality.

  Finally we approximate the idealized output state $\eta' = \Gamma_2(\sigma_1'\otimes\sigma_2')$ as $\wteta$ (with small support):
  \begin{align}
    \eta'&= \sum_{i,j} p_iq_j\left((1-c_{ij})\ketbra{\mathbf0}{\mathbf0}^{\otimes k} + c_{ij}\ketbra{\psi_i}{\psi_i}^{\otimes k}\right)\\
    \wteta &= \sum_{(i,j) \in S\wedge i\in S'} p_iq_j\left((1-c_{ij})\ketbra{\mathbf0}{\mathbf0}^{\otimes k} + c_{ij}\ketbra{\zeta_i}{\zeta_i}^{\otimes k}\right) + \sum_{(i,j) \notin S\vee i\notin S'} p_iq_j\ketbra{\mathbf0}{\mathbf0}^{\otimes k}
  \end{align}
  To bound $\norm{\eta'-\wteta}_1$, we need to bound three sources of error:
  (i) Accepting $(i,j) \notin S$, which occurs with probability $\Pr[A\cap S^c] \le \alpha\pacc\le\alpha$ by \cref{eq:Sc};
  (ii) Getting $(i,j) \in S$, but $i\notin S'$, which occurs with probability $\Pr[(i,j) \in S\wedge i\notin S'] \le \Pr[i\notin S'\mid (i,j) \in S]\le \gamma$;
  (iii) Approximation error $\sqrt{Ck/t}$ from \cref{eq:psik}.
  Therefore, we get
  \begin{equation}
    \begin{aligned}
      \norm{\eta'-\wteta}_1&\le \norm{\sum_{(i,j)\in S\wedge i\in S'}p_iq_j c_{ij}\left(\ketbra{\psi_i}{\psi_i}^{\otimes k} - \ketbra{\zeta_i}{\zeta_i}^{\otimes k}\right)}_1 + \norm{\sum_{(i,j)\notin S\vee i\notin S'} p_iq_j c_{ij}\left(\ketbra{\psi_{i}}{\psi_i}
      ^{\otimes k}-\ketbra{\mathbf0}{\mathbf0}^{\otimes k}\right)}_1 \\
      &\le \max_{i\in S'}\norm{\ketbra{\psi_i}{\psi_i}^{\otimes k} - \ketbra{\zeta_i}{\zeta_i}^{\otimes k}}_1 + 2\Pr[A \cap S^c] + 2\Pr[S\cap S^{\prime c}] \\ 
      &\le 2\left(\sqrt{Ck/t} + \alpha + \gamma\right).
    \end{aligned}
  \end{equation}
  We set 
  \begin{equation}
    t = Ck(8/\delta)^2, \quad \alpha = \gamma = \delta/16, \quad k' = \lceil-t\ln(\alpha\delta/4)\rceil,\quad \ell = \tilde O\!\left((k+k')^3/\delta^4\right).
  \end{equation}
  Thus, $\epsilon_\Lambda \le \delta/4$ by \cref{lem:disentangler}, and the overall error is
  \begin{equation}
    \norm{\eta-\wteta} \le 2\epsilon_\Lambda + 2\left(\sqrt{Ck/t} + \alpha + \gamma\right) \le \delta/2 + 2(\delta/8 + \delta/16+\delta/16) \le \delta.\qedhere
  \end{equation}
\end{proof}

Using the above lemma, we now show the containment $\pureQPH \subseteq \QPH$.
The main challenge is that a mixed state behaves like a probability distribution over pure states.
For instance, when Bob sends a mixed state proof $\rho = \sum_{i=1}^m p_i \ketbra{\psi_i}{\psi_i}$, Alice does not know which $\ket{\psi_i}$ the verifier will actually observe.
The disentangler of \cref{lem:pure-disentangler} resolves this issue.
If each proof is required to be product across four registers, then, after running the disentangler, the proof can be written in the form $\sum_{i=1}^m p_i \ket{\zeta_i}\!\!\bra{\zeta_i}^{\otimes k}$, where each $\ket{\zeta_i} = \ket{\zeta_{i,1}} \otimes \cdots \otimes \ket{\zeta_{i,s}}$, and, crucially, the number of terms $m$ is polynomially bounded. 

This structure means that Alice no longer needs to know which $\ket{\zeta_i}$ the verifier observes.
Instead, she can provide a response to every possible $\ket{\zeta_i}$ simultaneously, encoded in tensor product form.
Concretely, the $i$th message contains a table (in tensor product form) listing all possible transcripts of rounds $1,\dots,i-1$ together with Alice’s corresponding responses.
The disentangler ensures enough copies of each transcript are available, and the verifier can then use SWAP tests to select which transcript to use.

To enforce that each message is a product across four registers, we increase the number of rounds by a factor of $7$.
In other words, we show $\pureQSigmai[r] \subseteq \QSigmai[7r]$.
Within each block of seven quantifiers, we discard every other quantifier (positions $2,4,6$) and bundle the remaining four into a single quantifier ranging over product states:

\begin{equation}\label{eq:discard-quantifiers}
  \exists \rho_1\forall \rho_2\exists\rho_3\forall\rho_4\exists\rho_5\forall\rho_6\exists\rho_7 \quad\longmapsto\quad \exists \rho = (\rho_1\otimes\rho_3\otimes\rho_5\otimes\rho_7)
\end{equation}
Thus, by increasing the number of rounds by a factor of $7$, we simulate a proof system where the provers send states that are in tensor product across the four registers.

\amplification*

\begin{proof}
Let $A \in \pureQSigmai[r]$. Then there exist functions $c,s:\NN\to[0,1]$ with $c(n)-s(n) \ge n^{-O(1)}$, a polynomial $p(n)$, and a polynomial-time uniform family of verifiers $\{V_x\}_{x\in\{0,1\}^*}$ such that, for every $x \in \{0,1\}^n$, we have
  \begin{subequations}
    \begin{alignat}{2}
      x&\in \Ayes &\quad\Rightarrow\quad& \exists\ket{\psi_1}\forall\ket{\psi_2}\dotsm \overline{Q_r}\ket{\psi_{r-1}}Q_r{\ket{\psi_r}}\colon P_x(\ketbra{\psi_1}{\psi_1}\otimes\dotsm\otimes\ketbra{\psi_r}{\psi_r}) \ge c(n),\label{eq:Pxyes}\\
      x&\in \Ano &\quad\Rightarrow \quad& \forall\ket{\psi_1}\exists\ket{\psi_2}\dotsm Q_r\ket{\psi_{r-1}}\overline{Q_r}{\ket{\psi_r}}\colon P_x(\ketbra{\psi_1}{\psi_1}\otimes\dotsm\otimes\ketbra{\psi_r}{\psi_r}) \le s(n),
    \end{alignat}
  \end{subequations}
where $Q_r = \forall$ if $r$ is even and $Q_r = \exists$ if $r$ is odd.  
Here $P_x(\rho)$ denotes the acceptance probability of $V_x$ on input state $\rho \in \calD(\calH^{\otimes r})$, with $\calH = \CC^{2^{p(n)}}$.

We prove that $A \in \QSigmai[7r](c',s')$ by constructing a verifier $V_x'$ that receives $7r$ messages in $\calD(\calH^{\otimes \ell})$, where $\ell$ will be determined later (and depend on our application of \cref{lem:pure-disentangler}).  
As described in \cref{eq:discard-quantifiers}, the verifier $V_x'$ discards $3r$ of these messages and simulates an $r$-round protocol in which each round-$i$ message has the product form $\rho_i = \rho_{i,1}\otimes\dotsm\otimes\rho_{i,4}$.
Thus, it suffices to prove
  \begin{subequations}
    \begin{alignat}{2}
      x&\in \Ayes &\quad\Rightarrow\quad& \exists \rho_1\dotsm Q_r \rho_r\colon P'_x(\rho_1\otimes\dotsm\otimes\rho_r)\ge c'(n),\label{eq:P'yes}\\
      x&\in \Ano &\quad\Rightarrow \quad& \forall \rho_1\dotsm \overline{Q_r}\rho_r\colon P'_x(\rho_1\otimes \dotsm\otimes\rho_r) \le s'(n),\label{eq:P'no}
    \end{alignat}
  \end{subequations}
where $P'_x(\rho)$ denotes the acceptance probability of $V_x'$ and each $\rho_i$ is restricted to the four-register product form above.
Given $\rho_1, \dots, \rho_r \in \calD(\calH^{\otimes \ell})$ as input with $\rho_i = \rho_{i,1} \otimes \dots \otimes \rho_{i,4}$, the verifier $V_x'$ acts as follows:

\begin{enumerate}
\item (\emph{Determine the canonical transcript}) For $i=1,\dots,r$:
    \begin{enumerate}
    \item (\emph{Disentangle}) 
    Let $\Gamma$ denote the disentangler from \cref{lem:pure-disentangler}.  
    For the $i$th iteration, we write $\Gamma_i$, since each application will act on a different-sized Hilbert space.  
    Define $\rho_i' = \Gamma_i(\rho_i)$, where the parameters $K,\delta$ will be specified later.  
     By \cref{lem:pure-disentangler}, there exists a mixed state $\eta_i = \sum_{k=1}^{m_i}p_{ik}\ketbra{\zeta_{i}^{(k)}}{\zeta_i^{(k)}}^{\otimes K}$ with $m_i=O(\delta^{-2})$ such that $\norm{\eta_i - \rho_i'}_1 \le \delta$. 
      Additionally, each $\ket{\zeta_i^{(k)}}$ can be written as $\ket{\zeta_{i}^{(k)}} = \bigotimes_{j=1}^{M_i}\ket{T_{{ij}}^{(k)}}\ket{\psi_{ij}^{(k)}}$ for $M_i = M_{i-1}m_{i-1}$ (and $M_1 = 1$), $\ket{T_{ij}^{(k)}} \in \calH^{\otimes (i-1)}$, $\ket{\psi_{ij}^{(k)}} \in \calH$.
      For analysis, fix a pure branch $\ket{\zeta_{i}} = \bigotimes_{j=1}^{M_i}\ket{T_{{ij}}}\ket{\psi_{ij}}$ from this distribution (pretending the verifier $V'_x$ receives a random pure state from the mixed state).
      Each $\ket{T_{ij}}\ket{\psi_{ij}}$ is a transcript-answer pair; i.e., $\ket{\psi_{ij}}$ is the prover's message on the $i$th round conditioned on $\ket{T_{ij}}$ being the transcript for the previous $i-1$ rounds. Note that $M_i$ grows with each round because the transcript gets progressively longer each round.
      
      \item (\emph{Select player response to current transcript}) If $i=1$, set $\ket{\phi_1} \coloneq \ket{\psi_{i,1}}$ and $\ket{C_1} \coloneq \ket{\phi_{1}}$.
      Otherwise, for each $j=1,\dots,M_i$, perform $W$ (to be determined later) SWAP tests between $\ket{C_{i-1}}$ and $\ket{T_{ij}}$ (choose $K$ sufficiently large to have enough copies $\ket{C_{i-1}}$ and $\ket{T_{ij}}$).\footnote{Note that each $\ket{\zeta_i}$ contains a copy of each transcript and that we have $K$ copies of $\ket{\zeta_i}$.}
      If all SWAP tests accept for some $j$, set $\ket{C_i} \coloneq \ket{T_{ij}}\otimes\ket{\psi_{ij}}$.
      Else, let the current player lose, i.e., accept if $i$ is even and reject if $i$ is odd.
    \end{enumerate}
    \item Simulate $V_x$ on fresh copies of $\ket{C_r}$ $T$ times and accept if the number accepting runs, $\Nacc$, satisfies $\Nacc \ge T(c+s)/2$.
  \end{enumerate}

  \emph{Completeness.}
  Let $x\in \Ayes$.
  By \cref{eq:Pxyes}, Alice (first player, odd rounds) can always win with probability at least $c(n)$, regardless of which pure state Bob sends in the even rounds.
  Let $\ket{\alpha_1}$ be the best state Alice can choose in round $1$.
  For any Bob message $\ket{\beta_2}$ in round $2$, there exists $\ket{\alpha_3(\beta_2)}$, such that Alice can win with probability $\ge c(n)$.
  Inductively define $\ket{\alpha_i(\beta_2,\dots,\beta_{i-1})}$ as Alice's best response in round $i$, given Bob's messages $\ket{\beta_2},\dots,\ket{\beta_{i-1}}$ and Alice's messages $\ket{\alpha_1},\dots, \ket{\alpha_{i-2}(\beta_2,\dots,\beta_{i-3})}$.

  We need to show that \cref{eq:P'yes} holds.
  We can analyze $V_x'$ on the disentangled states as
  \begin{equation}\label{eq:Px''}
    \abs[\big]{P_x'(\rho_1\otimes\dotsm\otimes\rho_r) - P_x''(\eta_1\otimes\dotsm\otimes\eta_r)} \le r\delta/2,
  \end{equation}
  where $P_x''(\eta_1\otimes\dotsm\otimes\eta_r)$ denotes the acceptance probability of $V_x'$ when each $\rho_i'$ is replaced by $\eta_i$.
  For Alice's rounds, we can assume $\eta_i = \rho_i' = \Gamma_i(\rho_i)$ since Alice can always send a state of the correct product form.
  Further, Alice's answer in round $i$ may depend on $\eta_{1},\dots,\eta_{i-1}$, since $\rho_{j}'$ only depends on $\rho_j$, and the $\rho_j' \approx \eta_j$ approximation is merely an analytical tool and so Alice can choose any $\eta_i$ satisfying \cref{lem:pure-disentangler}.

  For $\rho_1$, Alice simply sends $4\ell$ copies of $\ket{\alpha_1}$.
  Now consider odd round $i>1$.
  There are $M_i = M_{i-1}m_{i-1}$ possible choices for the canonical transcript $\ket{C_{i-1}}$ (which includes Bob's message) after round $i-1$.
  Alice sends $4\ell$ copies of $\bigotimes_{j=1}^{M_i}\ket{T_{ij}}\ket{\alpha(T_{ij})}$, where $\ket{\alpha(T_{ij})}$ denotes Alice's best answer given transcript $\ket{T_{ij}}$ in the first $i-1$ rounds.
  This let's Alice pass the SWAP test in (1b) with probability $1$ (in the $P_x''(\eta_1\otimes\dotsm\otimes\eta_r)$ analysis with $\eta_i$ chosen by Alice).

  We argue that $V_x'$ always chooses a transcript that is accepted with probability almost $c$.
  There are two sources of error.
  The first is Bob cheating and altering the transcript, so that the verifier selects $\ket{T_{ij}}\ne \ket{C_{{i-1}}}$ in step (1b) of Bob's round.
  For $\td(\ketbra{T_{ij}}{T_{ij}}, \ketbra{C_{i-1}}{C_{i-1}})\le \epsilon$, Alice's chance of winning decreases by at most $\epsilon$. 
  If $\td(\ketbra{T_{ij}}{T_{ij}}, \ketbra{C_{i-1}}{C_{i-1}}) = \sqrt{1-\abs{\braket{C_{i-1}|T_{ij}}}^2} \ge \epsilon$, then
  $\abs{\braket{C_{i-1}|T_{ij}}}^2 \le 1-\epsilon^2$
  and the probability of all $W$ SWAP tests accepting is
  \begin{equation}\label{eq:prob-wrong-transcript}
    \pars*{\frac12 + \frac12\abs{\braket{C_{i-1}|T_{ij}}}^2}^W \le \pars*{1 - \frac{\epsilon^2}{2}}^W \le e^{-W\epsilon^2/2} \le \frac{1}{4qrM_r}
  \end{equation}
  for $W = \lceil2\epsilon^{-2}\ln (4qrM_r)\rceil$ with $\epsilon = \gamma/4r$ and $\gamma = c-s$.
  The second source of error is (1b) choosing a wrong $\ket{T_{ij}}$ for $\ket{C_i}$ in Alice's round.
  Alice does not lose in (1b), but there may be multiple $\ket{T_{ij}}$ close to $\ket{C_{i-1}}$.
  Again, Alice's winning probability decreases by at most $\td(\ketbra{T_{ij}}{T_{ij}}, \ketbra{C_{i-1}}{C_{i-1}})$, and the probability of choosing a ``bad'' transcript is bounded by \cref{eq:prob-wrong-transcript}.
  We can take the union bound over all rounds and entries in the tables to bound the probability that a bad transcript is selected by $1/(4q)$.
  Thus, Alice's winning probability decreases at most $r\epsilon=\gamma/4$ in total, which gives
  \begin{equation}
  \Pr[P_x(C_r) \ge c-\gamma/4] \ge 1-\frac{1}{4q},
  \end{equation}
  where the probability is taken over the choice of $\ket{\zeta_i^{(k)}}$ in step (1a) and outcome of the SWAP tests in (1b).
  Assuming $P_x(C_r)\ge c-\gamma/4$ and thus $\E[\Nacc]\ge (c-\gamma/4)T$, the probability of $V'_x$ rejecting in step 2 can be bounded with Hoeffding's inequality
  \begin{equation}
  \Pr[\Nacc\le (c-\gamma/2)T] \le \exp\pars*{-\frac{2(\gamma T/4)^2}{T}} \le \exp(-\gamma^2 T/8) \le \frac1{4q},
  \end{equation}
  for $T = \lceil8\gamma^{-2}\ln (4q)\rceil$.
  Setting $\delta = 1/(rq)$, and taking into account the disentangler error $1/(2q)$ of \cref{eq:Px''}, Alice wins with probability $\ge 1-1/q$.
  We have now assigned all parameters to polynomials in $n$.
  For all of the SWAP tests and simulations of $V_x$, we need $K \ge W\cdot rM_r + T$ copies of each message.
  Note $M_i$ grows exponentially in $r=O(1)$.

  \emph{Soundness.} For $x\in \Ano$ the analysis is analogous, just swapping the roles of Alice and Bob, i.e., `$\exists$' is now Bob.
  The only difference is that now the second player wins, which is insignificant for the above analysis.
\end{proof}

%% file: np_qph.tex
\section{Quantified Hamiltonian Complexity}

In this section, we initiate the study of quantified Hamiltonian problems. 
Our primary motivation is to identify complete problems for the various definitions of $\QPH$ to better understand these classes and the relationship among their different variants.
At the same time, quantified Hamiltonian problems are natural in their own right: they naturally generalize quantified Boolean formulae from classical complexity theory~\cite{stockmeyer1973word} to the quantum Hamiltonian setting.  
From a physical perspective, these problems capture robust versions of ground-state questions. For example, these problems allow us to ask: ``Does there exist a state on one subsystem such that, no matter how the rest of the system is perturbed, the total system remains in a low-energy state?''

\subsection{Quantified Hamiltonian Problems}

We now formally define the quantified Hamiltonian problems studied in this work. 
There are four natural variants we consider, determined by the following choices: (i) whether the quantified states are restricted to be pure or may be mixed, and (ii) whether the Hamiltonian is local or sparse.

\begin{definition}[$\exists \forall$-$k$-LH]
Let $2n$ be the number of qubits, $k \geq 1$ a fixed constant, and $a,b \in \R$ satisfying $b - a \geq 1/\poly(n)$. Given a $k$-local Hamiltonian $H$ as input, the task is to decide, under the promise that one of these holds:
\begin{itemize}
    \item (YES case): $\exists\rho\forall\sigma\colon \tr(H(\rho\otimes\sigma)) \le a$, or 
    \item (NO case): $\forall\rho\exists\sigma\colon \tr(H(\rho\otimes\sigma)) \ge b$.
\end{itemize}
We write $\exists\forall$-$k$-MLH (mixed, local Hamiltonian) for the version where $\rho$ and $\sigma$ may be mixed states and $\exists\forall$-$k$-PLH (pure, local Hamiltonian) for the version where $\rho = \ketbra{\psi}{\psi}$ and $\sigma = \ketbra{\phi}{\phi}$ are pure states.
\end{definition}

To move from local to sparse Hamiltonians, we replace the locality constraint with a sparsity condition on the input operator.

\begin{definition}[Row-sparse operators]\label{def:row-sparse}
   An operator $A$ is row-sparse if: 
   \begin{itemize}
       \item Each row of $A$ has at most $\poly(n)$ nonzero entries, and
       \item There exists a polynomial-time algorithm which, given a row index $i$, outputs the list of all pairs $(i, A_{ij})$ such that $A_{ij} \neq 0$.
   \end{itemize}
\end{definition}

We can now extend the quantified local Hamiltonian problems to their sparse-Hamiltonian counterparts. We define the problems below for an arbitrary constant number of quantifiers, as we will later prove completeness at every level.

\begin{definition}[Quantified sparse Hamiltonian problems]\label{def:PSH}
  Let $i\in\NN$ and fix a polynomial $q$.
  The promise problem $\MSHSigmai$ is defined as follows:
  \begin{itemize}
    \item (Input): A $d$-sparse Hamiltonian $H$ on $n=n_1+\dotsm+n_i$ qubits with $\maxnorm{H}\le q(n)$, thresholds $a,b\in\RR$ with $b-a\ge 1/q(n)$, and $H$ is defined by a circuit that given $r$ outputs all entries in row $r$ (see \cref{def:row-sparse}).\footnote{The parameters $d,n$ are implicitly bounded in terms of input size via the circuit description of $H$.}
    \item (YES case): $\exists \rho_1\forall \rho_2\dotsm Q_i\colon \tr(H(\rho_1\otimes\dotsm\otimes\rho_i))\le a$.
    \item (NO case): $\forall \rho_1 \exists \rho_2\dotsm \overline{Q_i} \rho_i\colon \tr(H(\rho_1\otimes\dotsm\otimes\rho_i))\ge b$.
  \end{itemize}
    Here, $Q_i$ is $\exists$ when $i$ is odd and $\forall$ when $i$ is even, and $\overline{Q_i}$ is the complementary quantifier.
    Each $\rho_j$ is quantified over $\calD(\calH_j)$ with $\calH_j = \CC^{2^{n_j}}$.

The pure variant $\PSHSigmai$ is defined identically, except that $\rho_1,\dots,\rho_i$ are restricted to pure states.
Finally, $\MSHPii$ and $\PSHPii$ are obtained by inverting all quantifiers.
\end{definition}

Generally, a problem is in $\pureQSigmai$ if and only if its complement is in $\pureQPii$.
Although $\PSHSigmai$ is not equal to the complement of $\PSHPii$, there is a trivial poly-time reduction.

\begin{lemma}\label{lem:complement}
  $\PSHPii \le_p \overline{\PSHSigmai}$ and $\PSHSigmai \le_p \overline{\PSHPii}$ for all $i\in\NN$, i.e., $\PSHPii$ is the complement of $\PSHSigmai$, up to poly-time many-one (aka Karp) reductions. The analogous statement holds for the mixed/sparse, pure/local, and mixed/local variants. 
\end{lemma}
\begin{proof}
  $(H,a,b)\in (\PSHPii)_\yes \;\Longleftrightarrow\; (-H,-b,-a)\in(\PSHSigmai)_{\no}$ follows directly from \cref{def:PSH}.
\end{proof}

In the remainder of this section, we establish containment and hardness results for the quantified Hamiltonian problems defined above.
The results we obtain for the two-quantifier versions are summarized in \cref{tab:hamiltonian-summary}; generalizing these to more quantifiers is relatively straightforward.

\begin{table}[ht]
\centering
\renewcommand{\arraystretch}{1.4} %
\setlength{\tabcolsep}{12pt} %
\begin{tabular}{c|c|c}
   & \textbf{Local} & \textbf{Sparse} \\ \hline
\textbf{Mixed} & $\in \NP^{\QMA} \cap \coNP^{\QMA}$ (\cref{cor:npqma-conpqma}) & $\in \QSigmai[2]$ (\cref{prop:msh}) \\ \hline
\textbf{Pure}  & $\in \NP^{\pureSuperQMA}$ (\cref{prop:puresuperqma}) & $\pureQSigmai[2]$-complete (\cref{thm:qph-completeness})
\end{tabular}
\caption{Variants of the $\exists\forall$-quantified Hamiltonian problems. All variants are contained in $\PSPACE$ 
except for the pure/sparse case, where the best known upper bound is $\NEXP$.}
\label{tab:hamiltonian-summary}
\end{table}

\subsection{Quantified Local Hamiltonian: \texorpdfstring{$\NP$ with a $\QMA$ Oracle}{NP with a QMA Oracle}}

We begin by analyzing the pure/local and mixed/local variants of the quantified Hamiltonian problem.
In particular, we show that these solved by oracle classes of the form $\NP^O$ for a promise class $O$, and we refer the reader to \cref{def:promise-oracle} for a definition of $\NP$ with an oracle to a promise class.

A key step in our proofs is checking the consistency of local density matrices: given a collection of reduced density matrices, does there exist a global quantum state that is consistent with all of them?
If the global state is allowed to be mixed, checking consistency is known to be $\QMA$-complete~\cite{Liu06}. If, instead, one must decide whether there exists a global pure state consistent with the reduced density matrices, the problem is $\pureSuperQMA$-complete~\cite{kamminga2025complexitypurestateconsistencylocal}.

Because the Hamiltonian is local, each term in $H$ acts on only a constant number of qubits. This means that for any purported proof state, it suffices for the prover to supply the reduced density matrices on just those local subsystems. The first witness for the quantified Hamiltonian problem can therefore be succinctly described by a classical list of local density matrices. The remaining task---verifying that these matrices are consistent with a true quantum state---can be outsourced to a $\QMA$ oracle.
We formalize this now.

\begin{proposition}\label{prop:existsforall}
  $\exists\forall$-$k$-$\mathrm{MLH} \in \NP^{\Vert \QMA[2]}$.
\end{proposition}

\begin{proof}
  Let $H = \sum_i H_i$ be the given $k$-local Hamiltonian.  
  The $\NP$ prover provides the collection of reduced density matrices of the candidate state $\rho$ on the supports of the local terms $H_i$.  

  First, the verifier checks that these reduced density matrices are \emph{consistent} with some global state. This can be done using a single $\QMA$ query, since consistency of local density matrices is $\QMA$-complete~\cite{Liu06}.  

  Next, for each term $H_i$, the verifier computes an effective operator 
  \[
    H_i' \;=\; \sum_j p_{ij}\,\bra{\psi_{ij}} H_i \ket{\psi_{ij}},
  \]
  where $\rho_i = \sum_j p_{ij} \ketbra{\psi_{ij}}{\psi_{ij}}$ is the reduced density matrix of $\rho$ on the qubits that $H_i$ acts upon.  
  Let $H' = \sum_i H_i'$. Then $H'$ is an operator acting only on the Hilbert space corresponding to the $\forall$-prover's state $\sigma$.  

  Finally, the verifier queries a $\coQMA$ oracle to check whether
  \[
    \forall \sigma \colon \tr(H' \sigma) \;\ge\; b.
  \]
  Since the procedure requires only one $\QMA$ query (for consistency) and one $\coQMA$ query (which can be implemented using $\QMA$), the entire protocol lies in $\NP^{\Vert \QMA[2]}$.
\end{proof}

It turns out that $\NP^{\Vert\QMA[2]} = \NP^\QMA$. 

\begin{proposition}\label{prop:NP^QMA}
  $\NP^{\QMA} \subseteq \NP^{\Vert\QMA[2]}$.
\end{proposition}
\begin{proof}
  Let $A\in \NP^\QMA$ be a promise problem.
  Consider an $\NP^{\QMA}$ verifier Turing machine $M$ for $A$.
  Construct $M'$ that only asks $2$ $\QMA$-queries.
  $M'$ receives as proof the original proof $y$ of $M$, as well as outcomes $z_1,\dots,z_m\in\{0,1\}$ to all queries of $M$.
  Assume the queries of $M$ are of the form $(H_i,a_i,b_i)$ for Hamiltonians $H_i$ with thresholds $a_i,b_i$.
  $M'$ can compute the $i$-th query by simulating $M$ with query answers $z_1,\dots,z_{i-1}$.
  $M'$ asks two $\QMA$ queries, one for all queries with $z_i=0$ and one for $z_i=1$.
  Let $m_i := (a_i+b_i)/2$.
  \begin{enumerate}
    \item Let $J_0 = \{i:z_i=0\}$.
    \begin{enumerate}
      \item[YES.] $\forall i\in J_0\colon \lmin(H_i) \ge m_i$.
      \item[NO.] $\exists i\in J_0\colon\lmin(H_i)\le a_i$.
    \end{enumerate}
    \item Let $J_1 = \{i:z_i=1\}$.
    \begin{enumerate}
      \item[YES.] $\forall i\in J_1\colon \lmin(H_i) \le m_i$.
      \item[NO.] $\exists i\in J_1\colon\lmin(H_i)\ge b_i$.
    \end{enumerate}
  \end{enumerate}
  Each can be done in a single query since $\QMA$ is closed under intersection and union.
  If both answers are YES, then $M'$ can fully simulate $M$.
  Otherwise $M'$ rejects.

  \emph{Completeness.}
  Let $x\in\Ayes$.
  Then $M$ has a robustly accepting branch (i.e., $M$ accepts regardless of the answers to invalid queries, see \cref{def:promise-oracle}).
  Assume the prover sends $z_i=1$ if $\lmin(H_i)\le m_i$ and $z_i=0$ otherwise.
  Then both queries are valid and the prover accepts.

  \emph{Soundness.}
  Let $x\in\Ano$.
  Then all branches of $M$ reject robustly.
  If the prover sends a wrong $z_i$ (to a valid query of $M$), then one of the two queries will reject and $M'$ rejects.
  So if both queries accept, $M'$ will still reject because $M$ rejects robustly.
\end{proof}

It is not clear whether a single query in \cref{prop:NP^QMA} suffices to simulate $\NP^\QMA$, because one query is a $\QMA$-query and the other is a $\coQMA$-query.

By \cref{lem:complement}, we immediately get an analogous result for $\forall\exists$-$k$-MLH.

\begin{proposition}\label{prop:forallexists}
$\forall\exists$-$k$-$\mathrm{MLH} \in \coNP^\QMA$.
\end{proposition}
\begin{proof}
Because $\exists\forall$-$k$-MLH is in $\NP^\QMA$, it's immediate that $\overline{\exists\forall\text{-$k$-MLH}}$ is in $\coNP^\QMA$. 
\cref{lem:complement} implies that there is a reduction from $\forall\exists$-$k$-MLH to $\overline{\exists\forall\text{-$k$-MLH}}$, which completes the proof.
\end{proof}

\begin{corollary}\label{cor:npqma-conpqma}
$\exists\forall$-$k$-$\mathrm{MLH} \in \NP^\QMA \cap \coNP^\QMA$.
\end{corollary}
\begin{proof}
By a min-max theorem (e.g., \cite[Theorem 2.2]{grewal_et_al:LIPIcs.CCC.2024.6}), we have that $\exists\forall$-MLH = $\forall\exists$-MLH. Thus, \cref{prop:existsforall,prop:forallexists} implies the result. 
\end{proof}

An argument essentially identical to that of \cref{prop:existsforall} yields the following containment for the pure/local case.

\begin{proposition}\label{prop:puresuperqma}
$\exists\forall$-$k$-$\mathrm{PLH} \in \NP^{\Vert\pureSuperQMA[2]}$.
\end{proposition}

Here, one query to the $\pureSuperQMA$ oracle verifies that the provided local density matrices are consistent with some global \emph{pure} state (a complete problem for $\pureSuperQMA$ \cite{kamminga2025complexitypurestateconsistencylocal}), and a second $\coQMA$ query is used exactly as in the proof of \cref{prop:existsforall}. We omit the details, since the argument carries over verbatim.

We remark that, unlike in the mixed-state case, we do not obtain containment in $\NP^\pureSuperQMA \cap \coNP^\pureSuperQMA$, because the minimax theorem invoked in \cref{cor:npqma-conpqma} does not apply when the proofs are restricted to pure states.

At present, there is no known approach to proving hardness for the intersection class $\NP^\QMA \cap \coNP^\QMA$. Moreover, such hardness results appear unlikely: we do not know of any complete problems for even $\NP \cap \coNP$ and if the $\exists\forall$-MLH problem were hard for either $\NP^\QMA$ or $\coNP^\QMA$, it would imply $\NP^\QMA = \coNP^\QMA$, an equality that seems implausible.

\subsection{Quantified Sparse Hamiltonian: Complete Problems for \texorpdfstring{$\pureQPH$}{pureQPH}}

We now turn to the \emph{sparse} variants of the quantified Hamiltonian problem.
Our first result establishes that the $\exists\forall$-mixed sparse Hamiltonian problem lies in $\QSigmai[2]$. 
More significantly, we prove that the quantified pure sparse Hamiltonian (PSH) problems are \emph{complete} for each level of the pure quantum polynomial hierarchy.
That is, for every $i \in \mathbb{N}$, $\PSHSigmai$ is $\pureQSigmai$-complete and $\PSHPii$ is $\pureQPii$-complete.
This gives the first natural family of complete problems for $\pureQPH$. 

\begin{proposition}\label{prop:msh}
    $\exists\forall$-$\mathrm{MSH}$ is contained in $\QSigmai[2]$. 
\end{proposition}
\begin{proof}
   The proof is identical to the containment result give in \cref{thm:qph-completeness} (below); we defer the details to that proof. 
\end{proof}

Our completeness result requires the following lemmas.
The first lemma is a simple but useful structural fact: if two registers are almost symmetric, then one register must be close to containing a copy of the other. This lets us ``pull out'' a clean copy of a state whenever the verifier enforces near-symmetry via a SWAP test.

\begin{lemma}\label{lem:sym-projection}
  Let $\ket{\psi}\in \calH_{A}$ and $\ket{\phi}\in \calH_{B}\otimes\calH_{C}$ with $\calH_A \cong \calH_{B}$, such that $\tr((\Pisym)_{AB} (\ketbra{\psi}{\psi}\otimes\ketbra{\phi}{\phi}))\ge 1-\epsilon>1/2$.
  Then there exists $\ket{\phi_2}\in \calH_{C}$, such that $\td(\ket{\phi},\ket{\psi,\phi_2}) \le \sqrt{2\epsilon}$.
\end{lemma}
\begin{proof}
  We have $\Pisym = \frac12(I+F)$, where $F$ is the SWAP operation on $AB$.
  It holds that
  \begin{equation}
    \begin{aligned}
    \tr(F(X\otimes Y)) &= \sum_{ij} \tr\bigl((\ketbra{i}{j}\otimes \ketbra{j}{i})(X\otimes Y)\bigr) = \sum_{ij} \tr(\ketbra{i}{j}X)\tr(\ketbra{j}{i}Y) \\
    &= \sum_{ij} \bra{j}X\ket{i}\bra{i}Y\ket{j} = \sum_{ij}\bra{j} XY\ket{j} = \tr(XY).
    \end{aligned}
  \end{equation}
  Let $\rho = \tr_{C}\ketbra{\phi}{\phi}$. 
  Thus,
  \begin{equation}
    \tr\bigl((\Pisym)_{AB} (\ketbra{\psi}{\psi}\otimes \ketbra{\phi}{\phi})\bigr) = \frac12\Bigl(1 + \tr\bigl(F(\ketbra{\psi}{\psi}\otimes \rho)\bigr)\Bigr) = \frac12(1 + \bra{\psi}\rho\ket{\psi}).
  \end{equation}
  Hence, $\bra{\psi}\rho\ket{\psi} \ge 1-2\epsilon$.
  Define
  \begin{equation}
    \ket{\phi_2} \coloneq \frac{(\bra{\psi}_B\otimes I_C)\ket{\phi}}{\sqrt{\bra\psi\rho\ket\psi}}.
  \end{equation}
  Therefore $\braket{\psi,\phi_2|\phi}  = \bra{\phi}(\ketbra{\psi}{\psi}\otimes I)\ket{\phi}/\sqrt{\bra\psi\rho\ket\psi} = \sqrt{\bra{\psi}\rho\ket{\psi}} \ge \sqrt{1-2\epsilon}$.
  Finally, $\td(\ket{\psi,\phi_2},\ket{\phi}) = \sqrt{1-\abs{\braket{\psi,\phi_2|\phi}}} \le \sqrt{2\epsilon}$.
\end{proof}

Next, we recall Kitaev’s circuit-to-Hamiltonian construction, which is the backbone of essentially all Hamiltonian complexity reductions. 

\begin{lemma}[{Kitaev's circuit-to-Hamiltonian mapping \cite{Kitaev2002}\protect\footnotemark}]\label{lem:kitaev}
  \footnotetext{The statement of this lemma is not explicitly made in \cite{Kitaev2002}. A direct proof can be found in \cite[Remark 3.3]{rudolph_et_al:LIPIcs.ITCS.2025.85}.}
  Let $V=U_m\dotsm U_1$ be a quantum circuit of $m$ $2$-local gates with $a$ ancilla qubits and $b$ input qubits.
  Then there exists a Hamiltonian $H^{(V)}$ that is the sum of $O(m)$ $5$-local projectors, such that 
  \begin{equation}
    \ker H = \spn\left\{\frac{1}{\sqrt{m+1}}\sum_{t=0}^m U_t\dotsm U_1(\ket{0^a}_{\calA}\ket{\psi}_\calB)\otimes\ket{1^t0^{m-t}}_{\calC}\middle|\ket{\psi}\in\CC^{2^b}\right\},
  \end{equation}
  where $\calA$ is the ancilla register, $\calB$ is the input register, and $\calC$ is the clock register.
  $H^{(V)}$ has a spectral gap of $\Omega(1/m^2)$.
\end{lemma}

The final tool we need is a variant of the well-known Projection Lemma, which frequently appears in Hamiltonian complexity.

\begin{lemma}[State Projection Lemma \cite{kamminga2025complexitypurestateconsistencylocal}]\label{lem:state-projection}
  Let $H = H_1 + H_2$ be the sum of two Hamiltonians acting on Hilbert space $\calH = \calS \oplus \calS^\perp$, where $\calS$ is the kernel of $H_2$ and the other eigenvalues are at least $J$.
  Let $\rho$ be a state in $\calH$ such that $\tr(H\rho) \le \epsilon$.
  Then there exists a state $\sigma$ (pure if $\rho$ is pure) in $\calS$, such that $\td(\rho,\sigma) \le \delta$ and $\tr(H\sigma) \le \epsilon + 2\delta\norm{H_1}$, for $\delta = \sqrt{(\epsilon+\norm{H_1})/J}$.
\end{lemma}

We are now ready to prove our completeness result.

\begin{theorem}\label{thm:qph-completeness}
  $\PSHSigmai$ is $\pureQSigmai$-complete and $\PSHPii$ is $\pureQPii$-complete.
\end{theorem}
\begin{proof}
  \emph{Containment.} $\PSHSigmai\in\pureQSigmai$ is completely analogous to the containment of the Separable Sparse Hamiltonian problem (i.e. $\exists\exists\mhyphen\mathrm{PSH}$) in $\QMA(2)$ \cite{chailloux2012complexity}.
  Given $d$-sparse $n$-qubit Hamiltonian $H$ with $0\preceq H\preceq I$ and error $\epsilon$, \cite{chailloux2012complexity} constructs a circuit $Q$ (using Hamiltonian simulation and phase estimation) that runs in time $\poly(d,n,\epsilon^{-1})$, such that for all states $\ket{\psi}$
  \begin{equation}
    \abs[\big]{\Pr[Q \text{ accepts }\ket{\psi}] - \bra{\psi}H\ket{\psi}} \le \epsilon.
  \end{equation}
  So all the $\pureQSigmai$-verifier needs to do is normalize the input Hamiltonian to satisfy $0\preceq H\preceq I$ and simulate $Q$ with $\epsilon = 1/(4q(n))$, which gives a promise gap of $1/(2q(n))$.

  \emph{Hardness.} Let $i\in\NN$. We will show $\PSHPii$ is $\pureQPii$-hard for even $i$, and $\PSHSigmai$ is $\pureQSigmai$-hard for odd $i$. The other two cases follow by \cref{lem:complement}.
  Let $i$ be even and $A\in \pureQPii$ (the proof for odd $i$ and $A\in \pureQSigmai$ is completely analogous).
  Given $x\in\{0,1\}^n$, we construct Hamiltonian $H_x$ in time $\poly(n)$, such that for $a,b$ to be determined later,
  \begin{subequations}
    \begin{alignat}{2}
      x&\in \Ayes &\quad\Rightarrow\quad& (H_x,a,b) \in (\PSHPii)_\yes\\
      x&\in \Ano &\quad\Rightarrow \quad& (H_x,a,b) \in (\PSHPii)_\no
    \end{alignat}
  \end{subequations}
  There exists a poly-time uniform family of verifiers $\{V_x\}$, such that 
  \begin{subequations}
    \begin{alignat}{2}
      x&\in \Ayes &\quad\Rightarrow\quad& \forall\ket{\psi_1}\exists\ket{\psi_2}\dotsm \forall\ket{\psi_{i-1}}\exists{\ket{\psi_i}}\colon P_x(\psi_1\otimes\dotsm\otimes\psi_i) \ge c(n)\label{eq:xAyes},\\
      x&\in \Ano &\quad\Rightarrow \quad& \exists\ket{\psi_1}\forall\ket{\psi_2}\dotsm \exists\ket{\psi_{i-1}}\forall{\ket{\psi_i}}\colon P_x(\psi_1\otimes\dotsm\otimes\psi_i) \le s(n)\label{eq:xAno},
    \end{alignat}
  \end{subequations}
  where $P_x(\rho)$ denotes the acceptance probability of $V_x$ on input $\rho$.
  Let $m \le n^{O(1)}$ be the number of gates of $V_x$ and $p<m$ (without loss of generality) the number of qubits in each message.
  Denote the $i$ message registers of $V_x$ by $\calB_1,\dots,\calB_i$ of $p$ qubits each.
  The Hamiltonian $H$ will act on registers $\calH_1=\calB_1,\dots,\calH_{i-1}=\calB_{i-1},\calH_{i}=\calA\calB\calC$, where $\calA$ is the ancilla register of $n_{\calA}\le m$ qubits, $\calB=\calB_1\dots\calB_i$ is the input register to $V$ of $ip$ qubits, and $\calC$ is the clock register of $m$ qubits.
  In terms of \cref{def:PSH}, we have $n_1=\dotsm=n_{i-1}=p$ and $n_i = ip+m+n_{\calA}$.
  Finally define the Hamiltonian
  \begin{equation}
    H_x = \ketbra{0}{0}_{\calA_1}\otimes\ketbra{1}{1}_{\calC_m} + J_1 \ketbra{0}{0}_{\calC_1}\otimes (I-\Pisym)_{\calH_1\dots\calH_{i-1},\calB_1\dots\calB_{i-1}} + J_2 H^{(V_x)}_{\calA\calB\calC},
  \end{equation}
  with $H^{(V_x)}$ from \cref{lem:kitaev}, $\Pisym$ the projector onto the symmetric subspace across the cut $\calH_1\dots\calH_{i-1}$ / $\calB_1\dots\calB_{i-1}$, and sufficiently large $1\ll J_1 \ll J_2 \le n^{O(1)}$.
  $H_x$ is $O(m)$-sparse since $H^{(V_x)}$ is local and has $O(m)$ terms, and $\Pisym$ is $2$-sparse.

  Let $a=(1-c)/(m+1)$ and $b=(1-c+\gamma/4)/(m+1)$, where $\gamma = c-s$.
  Note that the gap $b-a\ge 1/q(n)$ and bound $\maxnorm{H_x}\le O(J_2 m)\le q(n)$ for $n=n_1+\dotsm+n_i$ can be achieved by padding the last message (i.e. increasing $n_i$) and letting $H_x$ act as identity on the padding qubits.
  We lose purity of the last message, but that is not an issue since \cref{eq:xAyes,eq:xAno} are still true when replacing the pure $\ket{\psi_i}$ with a mixed $\rho_i$ (by convexity).

  For $x\in \Ayes$, we have \cref{eq:xAyes} we argue that
  \begin{equation}
    \forall \ket{\phi_1}\exists\ket{\phi_2}\dotsm\forall\ket{\phi_{i-1}}\exists\ket{\phi_i}\colon \tr\bigl(H(\phi_1\otimes\dotsm\otimes\phi_i)\bigr)\le a,
  \end{equation}
  holds with $\ket{\psi_j}\in\calH_j$ for all $j\in[i]$.
  For rounds $2,4,\dots,i-2$, Alice (taking the game interpretation with Bob as first player ($\forall$) and Alice as second player $(\exists)$) can simply send the same response as in \eqref{eq:xAyes}, i.e., $\ket{\phi_j} = \ket{\psi_j}$ for even $j <i$.
  In the last round, Alice sends a valid history state of the form 
  \begin{equation}
    \ket{\phi_i} = \frac1{\sqrt{m+1}}\sum_{t=0}^m U_{t}\dotsm U_1(\ket{0}_{\calA}\ket{\psi_1,\dots,\psi_i}_{\calB})\otimes\ket{1^t0^{m-t}}_{\calC},
  \end{equation}
  where $\ket{\psi_i}$ corresponds to Alice's last message in \eqref{eq:xAyes}, and $U_1,\dots,U_m$ are the gates of $V_x$ with output register $\calA_1$.
  Then we get $\tr(H_x(\phi_1\otimes\dotsm\otimes\phi_i)) \le (1-c)/(m+1)$ since $V_x$ rejects $\ket{\psi_1,\dots,\psi_i}$ with probability $\le 1-c$.

  Now consider $x\in\Ano$ and assume
  \begin{equation}\label{eq:contradiction}
    \forall \ket{\phi_1}\exists\ket{\phi_2}\dotsm\forall\ket{\phi_{i-1}}\exists\ket{\phi_i}\colon \tr\bigl(H(\phi_1\otimes\dotsm\otimes\phi_i)\bigr) < b.
  \end{equation}
  We will show that this contradicts \cref{eq:xAno}.
  Let $\ket\phi = \ket{\phi_1,\dots,\phi_{i}}$, and $\epsilon = \gamma/(8(m+1))$.
  By \cref{lem:kitaev,lem:state-projection} and choosing sufficiently large $J_2\in\poly(J_1,m,\gamma^{-1})$, there exists a state $\ket{\phi_i'}$, such that for some input state $\ket{\eta}$:
  \begin{equation}
    \begin{aligned}
    &\ket{\phi_i'} = \frac{1}{\sqrt{m+1}}\sum_{t=0}^m U_t\dotsm U_1 \ket{0}_\calA\ket{\eta}_\calB\otimes\ket{1^t0^{m-t}}_{\calC},\\
    &\quad\td\bigl(\ket{\phi_i},\ket{\phi_i'}\bigr)\le \epsilon, \quad \text{and }\;\tr\bigl(J_1(I-\Pisym)(\phi_1\otimes\dotsm\otimes\phi_{i-1}\otimes\phi_{i}')\bigr) \le b + J_1\epsilon.
    \end{aligned}
  \end{equation}
  By \cref{lem:sym-projection} and choosing sufficiently large $J_1$, there exists a state $\ket{\eta_i}$, such that
  \begin{equation}\label{eq:eta}
  \td(\ket{\phi_1,\dots,\phi_{i-1},\eta_i},\ket{\eta}) \le \epsilon,
  \end{equation}
  and therefore $\td(\ket{\phi_i},\ket{\phi_i''}) \le 2\epsilon$, with
  \begin{equation}
    \ket{\phi_i''} = \frac{1}{\sqrt{m+1}}\sum_{t=0}^m U_t\dotsm U_1 \ket{0}_\calA\ket{\phi_1,\dots,\phi_{i-1},\eta_i}_\calB\otimes\ket{1^t0^{m-t}}_{\calC}.
  \end{equation}
  Hence,
  \begin{equation}
    \tr\bigl(H(\phi_1\otimes\dotsm\otimes\phi_{i-1}\otimes\phi_{i}'')\bigr) = \frac{1}{m+1}\bigl(1-P_x(\phi_1\otimes\dotsm\otimes\phi_{i-1}\otimes\eta_i)\bigr) \le b + 2\epsilon.
  \end{equation}
  Thus, $P_x(\phi_1\otimes\dotsm\otimes\phi_{i-1}\otimes\eta_i) \ge c-\gamma/4-\gamma/4 =c-\gamma/2=s+\gamma/2$.
  Therefore,
  \begin{equation}
    \forall \ket{\psi_1}\exists\ket{\psi_2}\dotsm\forall\ket{\psi_{i-1}}\exists\ket{\eta_i}: P_x(\psi_1\otimes\dotsm\otimes\psi_{i-1}\otimes\eta_i) \ge s+\gamma/2,
  \end{equation}
  which means that Alice can win with probability at least $s+\gamma/2$ by choosing the even messages $\ket{\psi_2}=\ket{\phi_{2}},\dots,\ket{\psi_{i-2}}=\ket{\phi_{i-2}},\ket{\psi_i}=\ket{\eta_i}$ with $\ket{\phi_1},\dots,\ket{\phi_i}$ as in \cref{eq:contradiction}, and $\ket{\eta_i}$ as in \cref{eq:eta}.
  This contradicts \cref{eq:xAno}.
\end{proof}

We note that proving that $\exists\forall$-MSH is $\QSigmai[2]$-hard seems out of reach with current techniques. 
\cref{thm:qph-completeness} relies on the SWAP test, but this approach fails in the mixed-state setting, because the SWAP test fails to check equality of mixed states. 
Any hardness proof would therefore require fundamentally new ideas.